\documentclass[journal]{IEEEtran}
\ifCLASSINFOpdf
\else
\fi

\usepackage[cmex10]{amsmath}

\usepackage{amssymb,amsthm}
\usepackage{color}

\theoremstyle{definition}
\newtheorem{definition}{Definition}

\newtheorem{example}[definition]{Example}

\theoremstyle{plain}
\newtheorem{theorem}{Theorem}
\newtheorem{proposition}[definition]{Proposition}
\newtheorem{lemma}[definition]{Lemma}
\newtheorem{remark}[definition]{Remark}
\newtheorem{corollary}[definition]{Corollary}

\begin{document}
%
\title{Communication efficient and strongly secure secret sharing schemes based on algebraic geometry codes}
%
%
%
\author{Umberto~Mart{\'i}nez-Pe\~{n}as,~\IEEEmembership{Student Member,~IEEE}
\thanks{The work of U.~Mart{\'i}nez-Pe\~{n}as was supported by The Danish Council for Independent Research under Grant No. DFF-4002-00367 and Grant No. DFF-5137-00076B (``EliteForsk-Rejsestipendium'').

U. Mart{\'i}nez-Pe\~{n}as was with the Department of Mathematical Sciences, Aalborg University, Aalborg 9220, Denmark. He is now with The Edward S. Rogers Sr. Department of Electrical and Computer Engineering, University of Toronto, Toronto, ON M5S 3G4 Canada. (e-mail: umberto@math.aau.dk).}}


%
%

\markboth{}%
{Shell \MakeLowercase{\textit{et al.}}: Bare Demo of IEEEtran.cls for IEEE Journals}
%



\maketitle

\begin{abstract}
Secret sharing schemes with optimal and universal communication overheads have been obtained independently by Bitar et al. and Huang et al. However, their constructions require a finite field of size $ q > n $, where $ n $ is the number of shares, and do not provide strong security. In this work, we give a general framework to construct communication efficient secret sharing schemes based on sequences of nested linear codes, which allows to use in particular algebraic geometry codes and allows to obtain strongly secure and communication efficient schemes. Using this framework, we obtain: 1) schemes with universal and close to optimal communication overheads for arbitrarily large lengths $ n $ and a fixed finite field, 2) the first construction of schemes with universal and optimal communication overheads and optimal strong security (for restricted lengths), having in particular the {\color{black}component-wise} security advantages of perfect schemes and the {\color{black}security and} storage efficiency of ramp schemes, and 3) schemes with universal and close to optimal communication overheads and close to optimal strong security defined for arbitrarily large lengths $ n $ and a fixed finite field. \\
\end{abstract}

\begin{IEEEkeywords}
Secret sharing, algebraic geometry codes, communication efficiency, communication bandwidth, strong security, asymptotic secret sharing.
\end{IEEEkeywords}

%
\IEEEpeerreviewmaketitle

\section{Introduction} \label{sec intro}
%
%
%
%
\IEEEPARstart{A}{ secret} sharing scheme \cite{blakley-safe, shamir} is a procedure to encode a secret into $ n $ shares, given bijectively to $ n $ parties, in such a way that some specified collections of shares give all the information about the secret, while other collections of shares give no information at all. It is usual to specify such collections by threshold values $ t < r $: the secret can be recovered from any $ r $ shares, while no collection of $ t $ shares gives any information about the secret.

Traditionally, efforts have been made to obtain schemes where $ r $ is as low and $ t $ is as large as possible. The first construction, by Shamir \cite{shamir}, consists of a secret sharing scheme attaining information rate $ 1 / n $ (perfect scheme) and any threshold values $ 0 \leq t < r \leq n $ with $ t = r-1 $, which are optimal. Later, this construction was extended in \cite{blakley, yamamoto-threshold} to higher information rates $ \ell / n $ and threshold values $ t = r - \ell $, which are again optimal. Schemes where the information rate is $ \ell/n > 1/n $ (non-perfect schemes) are also called ramp schemes, and allow the shares to be smaller than the secret, hence allowing more flexibility and efficiency in terms of storage, at the cost of a lower privacy threshold $ t < r-1 $. All these previous optimal constructions require a finite field of size $ q > n $, which in many situations requires performing computations over large fields, reduces the storage efficiency since the shares are larger, {\color{black}and provides less robustness against erasures (see Subsection \ref{subsec our motivations})}. 

{\color{black}In the works \cite{massey, mceliece} and later in} \cite{secure-computation, duursma, kurihara-secret}, more general frameworks for constructing secret sharing schemes are given in terms of linear (block) codes. Thanks to these approaches, Shamir's scheme is extended to schemes based on algebraic geometry codes \cite{CascudoMPC, chen, Cramer-Cascudo, secure-computation, duursma, cramerbook}, which may have arbitrarily large length $ n $ for any fixed finite field at the cost of achieving close to optimal parameters, instead of optimal, being the difference dependent only on the field size. This implies that, for a given difference with respect to the optimal case, we may fix a suitable field size and the achievable lengths are not restricted.

On the other hand, as noticed first in \cite{wangwong} and then in \cite{efficient}, the overall amount of communication between the contacted parties and the user (decoding bandwidth), in a given secret sharing scheme, can be decreased by contacting more than $ r $ shares. A lower bound on the communication overhead (the difference between the decoding bandwidth and the information contained in the secret) was first given in \cite{wangwong} for perfect schemes, and then in \cite{efficient} for the general case. 

A modification of Shamir's secret sharing scheme was given independently in \cite{rawad} and in \cite{efficient} for any information rate, with optimal threshold values $ r $ and $ t $, and where optimal communication overheads are simultaneously achieved (universal) for any $ \delta $ contacted shares and any {\color{black}$ \delta $ in an arbitrary set $ \Delta \subseteq [r,n] $}. However, as Shamir's scheme, this construction requires a finite field of size $ q > n ${\color{black}, which in this case is more critical since shares are exponentiated by a number $ \alpha $ of order at least $ n $ for fixed information rates (see Subsection \ref{subsec our motivations}).}

Therefore, the following question arises naturally: can we modify or extend the constructions in \cite{rawad, efficient} so that we may apply algebraic geometry codes and hence obtain near optimal thresholds and universal near optimal communication overheads for arbitrarily large lengths $ n $ and a fixed finite field? 

In this work, we answer the previous question positively, giving a general framework for constructing communication efficient secret sharing schemes based on sequences of nested linear codes in a similar way to that of \cite{secure-computation, duursma, kurihara-secret, massey, mceliece}. Afterwards and using algebraic geometry codes, we obtain the desired communication efficient secret sharing schemes of arbitrary length and fixed finite field. More concretely, thanks to our general framework, we reduce the problem of finding long communication efficient secret sharing schemes to the well-known problem of finding algebraic curves with many rational points and low genus \cite{ihara, stichtenothbook, vladut}. {\color{black}We highlight here that such secret sharing schemes can be constructed explicitly whenever the generator matrices of the corresponding algebraic geometry codes are known. This is the case of the families in \cite{notehermitian, tiersma, vanlint-springer} and \cite{garciatower, shum}. }

Again in this case, the difference with the lowest communication overheads depends only on the field size, and hence for a given defect with respect to the optimal value, we may fix a suitable field size, and the achievable lengths are not restricted. Moreover, as we will show, one can see the defect with respect to the optimal case just on the privacy threshold $ t $, as in previous schemes based on algebraic geometry codes, meaning that the obtained communication overheads give the same improvement on the reconstruction threshold $ r $ as in the optimal case {\color{black}(see Example \ref{example hermitian parameters})}.

The introduced universal construction will {\color{black}expand the} alphabet size {\color{black}in a similar way to} that of the universal constructions in \cite{rawad, efficient}. As in \cite{rawad}, we will also give a non-universal construction that has a significant smaller alphabet. 

On the other hand, our general framework also allows us to obtain schemes that are communication efficient and strongly secure at the same time. Strongly secure schemes were introduced in \cite{yamamoto-threshold}, and allow to keep {\color{black}components} of the secret safe even when more than $ t $ shares are eavesdropped. In such a scheme, if the secret is constituted by $ \ell $ components (each considered as a secret by itself), where $ \ell / n $ is the information rate, then no information about any tuple of $ \ell - \mu $ components is leaked if $ t + \mu $ shares are eavesdropped. In particular, in an optimal strongly secure ramp scheme with information rate $ \ell / n $ and also optimal thresholds $ t = r - \ell $, no information about any component of the secret is leaked by eavesdropping any $ r-1 $ shares, as in Shamir's optimal perfect scheme {\color{black}(although information about the whole secret may still leak)}, while having the storage advantages of optimal ramp schemes (since it actually is an optimal ramp scheme). In conclusion, an optimal strongly secure ramp scheme has the {\color{black}component-wise} security advantages of perfect schemes and the {\color{black}security and} storage efficiency of ramp schemes.

By means of the framework developed in this paper, we obtain the first construction of a secret sharing scheme with universal and optimal communication efficiency and optimal strong security at the same time. However, it requires a finite field of size $ q > \ell + n $, as previous constructions of optimal strongly secure schemes \cite{nishiara}. As before, to relax this requirement, we give another construction based on algebraic geometry codes, whose communication efficiency and strong security are close to the optimal values (again depending on $ q $ and not on $ n $), while having arbitrarily large lengths $ n $ over a fixed finite field. {\color{black}Such strongly secure schemes based on algebraic geometry codes were obtained in \cite{ryutaroh-multi} without communication efficiency.}

For convenience of the reader, a comparison between the obtained constructions in this paper (Corollaries \ref{corollary SSS one-point construction2}, \ref{corollary strongly secure MDS} and \ref{corollary strongly secure AG}) and the ones in the literature is summarized in Fig. \ref{fig comparison}.

\begin{figure*}[!t]
\centering
\begin{tabular}{ccccc}
\hline
Construction & Thresholds vs info rate & Comm. efficiency & Strong security & Length vs field \\
\hline \hline
Shamir \cite{shamir} & Optimal (perfect) & No & No & $ n < q $ \\
\hline 
\cite{blakley, yamamoto-threshold} & Optimal (ramp) & No & No & $ n < q $ \\
\hline 
\cite{chen, secure-computation} & Near optimal & No & No & $ n \gg q $ \\
\hline 
\cite{rawad, efficient} & Optimal & Yes & No & $ n < q $ \\
\hline 
Corollary \ref{corollary SSS one-point construction2} & Near optimal & Yes & No & $ n \gg q $ \\
\hline \hline
{\color{black}\cite{nishiara, yamamoto-threshold}} & Optimal & No & Yes & $ \ell + n < q $ \\
\hline 
{\color{black}\cite{ryutaroh-multi}} & Near optimal & No & Yes & $ \ell + n \gg q $ \\
\hline 
Corollary \ref{corollary strongly secure MDS} & Optimal & Yes & Yes & $ \ell + n < q $ \\
\hline
Corollary \ref{corollary strongly secure AG} & Near optimal & Yes & Yes & $ \ell + n \gg q $ \\
\hline
\end{tabular}
\caption{Constructions in this work vs those in the literature{\color{black}, where $ \ell/n $ is the information rate, thus $ n < \ell +n \leq 2n $}.}
\label{fig comparison}
\end{figure*}

{\color{black}
\subsection{Motivation and meaning of our contributions} \label{subsec our motivations}

As mentioned at the beginning of the Introduction, the optimal schemes in \cite{rawad, efficient} require that $ q = \Omega(n) $, whereas those based on algebraic geometry codes may have arbitrarily large lengths $ n $ for a fixed field size $ q $. Here we develop why this relaxed condition may be crucial in the applications:

1) Secure multi-party computation \cite{rawadMPC, CascudoMPC, chen, Cramer-Cascudo, secure-computation, MPCfromanySS, cramerbook} allows to securely compute any function \cite{ben-or, chaum}. Such functions are usually defined over small fields and computing them requires several rounds of operations. Hence a field size $ q = \Omega(n) $ may be inefficient and restrictive for large number of parties $ n $.

The schemes in \cite{rawad, efficient} and this work are linear, which allows to securely compute any linear function or any binary function, and is of interest, for instance, in additively homomorphic verifiable secret sharing. Our schemes are multiplicative for certain parameters. See Appendix \ref{app multiplicative} for details.
 
2) The known communication efficient schemes \cite{rawad, efficient} require exponentiating secret and share sizes by a large number $ \alpha $, which reduces storage flexibility. In the best of scenarios (Construction 1), for fixed information rate $ L $ we have that $ \alpha = \lambda n $ for a constant $ L < \lambda < 1 $. Thus each share in \cite{rawad, efficient} has size at least $ n^{\lambda n} $, that is $ \Omega(n \log(n)) $ bits. In this case, our schemes based on asymptotically good algebraic geometry codes allow the secret and share sizes to be $ O(q^{\lambda n}) $, for some $ L < \lambda < 1 $ and small $ q $. Hence each share is of $ O(n) $ bits. See the observations at the end of Subsection \ref{subsec asymptotic}.

3) On most erasure channels, increasing code lengths (number of participants) while keeping or decreasing symbol sizes (share sizes) provides more robustness even if the allowed fraction of erasures is lower. This is even more critical on adversarial scenarios where the adversary chooses which fraction of symbols to erase. For instance assume that we have a scheme with length $ 1000 $ and symbol size $ 2^{16} $, and a second scheme with length $ 4000 $ and symbol size $ 2^4 $. If the adversary erases $ 1/4 $ of symbols of size $ 2^4 $, he or she may make reconstruction impossible when using the first scheme, even if it has optimal reconstruction threshold, while the second scheme may recover the original data without problems. The fact that share sizes can be of $ O(n) $ bits in contrast with $ \Omega(n\log(n)) $ bits, makes our schemes more robust in this sense when compared to the schemes in \cite{rawad, efficient}. See also \cite[Example VI]{vanlint-springer} and Figure \ref{fig hermitian} for a comparison of some explicit parameters when using Hermitian codes.

4) Finally, the strong security added in Section \ref{sec strongly secure} to the previous constructions also helps alleviate the storage requirements of perfect schemes, while keeping their component-wise security.

}

{\color{black}
\subsection{Organization of the paper} \label{subsec organization}
}

After some preliminaries from the literature in Section \ref{sec preliminaries}, the results in this paper are organized as follows: In Section \ref{sec constructions}, we give two general constructions of communication efficient secret sharing schemes based on linear codes, which are inspired by the staircase presentation in \cite{rawad}. Construction 1 aims at obtaining low communication overheads for sets of size $ \delta $ for one fixed value $ r \leq \delta \leq n $, while Construction 2 aims at obtaining low communication overheads for sets of size $ \delta $, for any {\color{black}$ \delta $ in an arbitrary set $ \Delta \subseteq [r,n] $,} simultaneously, at the cost of a larger alphabet than Construction 1. In Section \ref{sec AG codes}, we specialize the previous linear codes to algebraic geometry codes, and see that the resulting schemes have close to optimal communication overheads, while their lengths are not bounded by the field size. In Section \ref{sec strongly secure}, we give special cases of Constructions 1 and 2 from Section \ref{sec constructions} that aim at achieving strong security. By specializing the linear codes to MDS codes (such as Reed-Solomon codes), we obtain the first construction of schemes with universal optimal communication efficiency and optimal strong security at the same time, at the cost of a large field size $ q $ (but still of the order of $ n $). On the other hand, by specializing the linear codes to algebraic geometry codes, we obtain schemes with close to optimal communication efficiency and strong security, while having lengths not bounded by the field size.

\section{Definitions and preliminaries} \label{sec preliminaries}

\subsection{Notation} \label{subsec notation}

Throughout the paper, $ q $ denotes a fixed prime power, and $ \mathbb{F}_q $ denotes the finite field with $ q $ elements. For positive integers $ m $ and $ n $, we denote by $ \mathbb{F}_q^n $ the vector space of row vectors of length $ n $ over $ \mathbb{F}_q $, and we denote by $ \mathbb{F}_q^{m \times n} $ the vector space of $ m \times n $ matrices with components in $ \mathbb{F}_q $. 

We also denote $ [n] = \{ 1,2, \ldots, n \} $ and $ [m, n] = \{ m, m+1, \ldots, n-1, n \} $ if $ m \leq n $. For any vector $ \mathbf{x} \in \mathbb{F}_q^n $, any matrix $ X \in \mathbb{F}_q^{m \times n} $ and any set $ I \subseteq [n] $, we denote by $ \mathbf{x}_I $ and $ X_I $ the vector in $ \mathbb{F}_q^{\# I} $ and the matrix in $ \mathbb{F}_q^{m \times \# I} $ obtained by restricting $ \mathbf{x} $ to the components indexed by $ I $ and restricting $ X $ to the columns indexed by $ I $, respectively.

Sets $ \mathcal{C} \subseteq \mathbb{F}_q^n $ and $ \mathcal{D} \subseteq \mathbb{F}_q^{m \times n} $ are called linear codes if they are vector spaces over $ \mathbb{F}_q $. We denote by $ \mathcal{C}^\perp $ the dual of the code $ \mathcal{C} $ with respect to the usual inner product in $ \mathbb{F}_q^n $ and, for a set $ I \subseteq [n] $, we define the restriction of $ \mathcal{C} $ and $ \mathcal{D} $ to $ I $, respectively, as the linear codes
$$ \mathcal{C}_I = \{ \mathbf{c}_I \mid \mathbf{c} \in \mathcal{C} \} \subseteq \mathbb{F}_q^{\# I} \quad \textrm{and} $$
$$ \mathcal{D}_I = \{ D_I \mid D \in \mathcal{D} \} \subseteq \mathbb{F}_q^{m \times \# I}, $$
{\color{black}and we define their shortened codes in $ I $ as the linear codes
$$ \mathcal{C}^I = \{ \mathbf{c}_I \mid \mathbf{c} \in \mathcal{C}, \mathbf{c}_{[n] \setminus I} = \mathbf{0} \} \subseteq \mathbb{F}_q^{\# I} \quad \textrm{and} $$
$$ \mathcal{D}^I = \{ D_I \mid D \in \mathcal{D}, D_{[n] \setminus I} = 0 \} \subseteq \mathbb{F}_q^{m \times \# I}. $$}
Finally, for random variables $ X $ and $ Y $, we denote by $ H(X) $, $ H(X | Y) $ and $ I(X; Y) $ the entropy of $ X $, conditional entropy of $ X $ given $ Y $, and the mutual information between $ X $ and $ Y $, respectively, taking logarithms with some specified base. See \cite[Chapter 2]{cover}.

\subsection{Communication efficient secret sharing}

We will consider the general definition of secret sharing schemes from \cite[Definition 1]{efficient}. 

\begin{definition} [\textbf{Secret sharing schemes \cite{efficient}}]
Let $ \mathcal{A} $ be an alphabet, and let $ n $, $ \ell $, $ r $ and $ t $ be positive integers. An $ (n, \ell, r, t)_\mathcal{A} $ secret sharing scheme is a randomized encoding function $ F : \mathcal{A}^\ell \longrightarrow \mathcal{A}^n $ (meaning that $ F(\mathbf{s}) $ may take several values in $ \mathcal{A}^n $ with a given probability distribution, for each $ \mathbf{s} \in \mathcal{A}^\ell $) such that:
\begin{enumerate}
\item
It has $ r $-reconstruction: For any secret $ \mathbf{s} \in \mathcal{A}^\ell $, if $ \mathbf{x} = F(\mathbf{s}) $ and $ I \subseteq [n] $ is of size at least $ r $, then 
$$ H(\mathbf{s} | \mathbf{x}_I) = 0. $$
\item
It has $ t $-privacy: For any secret $ \mathbf{s} \in \mathcal{A}^\ell $, if $ \mathbf{x} = F(\mathbf{s}) $ and $ I \subseteq [n] $ is of size at most $ t $, then 
$$ H(\mathbf{s} | \mathbf{x}_I) = H(\mathbf{s}). $$
\end{enumerate}
We define the information rate of the scheme as $ \ell / n $.
\end{definition}

This definition includes most of the classical proposals for secret sharing schemes \cite{rawad, blakley-safe, blakley, CascudoMPC, chen, Cramer-Cascudo, secure-computation, MPCfromanySS, duursma, one-point, efficient, kurihara-secret, massey, ryutaroh-multi, mceliece, nishiara, cramerbook, shamir, wangwong, yamamoto-threshold}. As usual in the literature, we identify the $ i $-th share $ x_i $ with the $ i $-th party, hence the set $ [n] $ represents both the $ n $ parties and their corresponding $ n $ shares. 

Traditional efforts have been made in order to obtain schemes where $ r $ is as low as possible, while keeping $ t $ and $ \ell $ large with respect to $ n $. The obvious benefit of a low value of $ r $ is that less parties need to be contacted, hence it is expected that the amount of transmitted information to recover the secret is lower. However, as noticed first in \cite{wangwong} and later in \cite{efficient}, it is possible to recover the secret while transmitting less information, as long as more parties are contacted. Observe that, in the previous definition, when contacting $ r $ parties to recover the secret, we use their whole shares. However if more parties are contacted, then we may preprocess their shares so that the overall amount of transmitted information is lower. 

Formally, let $ I \subseteq [n] $ be such that $ \# I \geq r $, let $ i \in I $, and let $ E_{I,i} : \mathcal{A} \longrightarrow \mathcal{B}_{I,i} $ be preprocessing functions, where $ \# \mathcal{B}_{I,i} \leq \# \mathcal{A} $. The existence of decoding functions $ D_I : \prod_{i \in I} \mathcal{B}_{I,i} \longrightarrow \mathcal{A}^\ell $ such that $ D_I((E_{I,i}(x_i))_{i \in I}) = \mathbf{s} $, where $ \mathbf{x} = F(\mathbf{s}) $, for all $ \mathbf{s} \in \mathcal{A}^\ell $, is equivalent to
\begin{equation}
H(\mathbf{s} | (E_{I,i}(x_i))_{i \in I}) = 0,
\label{eq condition preprocessing functions}
\end{equation}
for all $ \mathbf{s} \in \mathcal{A}^\ell $. {\color{black}Observe that each preprocessing function $ E_{I,i} $ may depend on $ i $ and $ I $. However, we will simply write $ E $ when $ i $ and $ I $ are understood from the context.}

We may now define the communication overhead and decoding bandwidth as follows:

\begin{definition} [\textbf{Communication overhead and decoding bandwidth \cite{efficient}}]
For a set $ I \subseteq [n] $ and preprocessing functions satisfying (\ref{eq condition preprocessing functions}), we define the communication overhead and decoding bandwidth of $ I $, respectively, as
$$ {\rm CO}(I) = \sum_{i \in I} H(E(x_i)) - H(\mathbf{s}), \quad \textrm{and} $$
$$ {\rm DB}(I) = \sum_{i \in I} H(E(x_i)), $$
{\color{black}where logarithms in the entropy functions are with base $ \# \mathcal{A} $.}
\end{definition}

Observe that, after preprocessing the $ i $-th share, the amount of information required to be transmitted by the $ i $-th party is $ H(E(x_i)) $. Without assuming any processing of the overall information in the variables $ E(x_i) $, for $ i \in I $, the total transmitted information by the contacted parties is just the sum of $ H(E(x_i)) $ for $ i \in I $. Hence the decoding bandwidth is the total amount of information transmitted by the parties indexed by $ I $. Observe that
$$ {\rm DB}(I) \geq H((E(x_i))_{i \in I}) \geq H(\mathbf{s}), $$
where the first inequality is a particular case of \cite[Theorem 2.6.6]{cover} and the second inequality is a particular case of \cite[Exercise 2.4]{cover}. Therefore, the communication overhead measures the amount of overall extra information transmitted by the contacted parties.

From now on, we will assume that the secret $ \mathbf{s} $ is a uniform random variable on $ \mathcal{A}^\ell $. To measure the quality of a scheme, we will use the following bounds given in \cite[Proposition 1]{efficient} and \cite[Theorem 1]{efficient}:

\begin{proposition} [\textbf{\cite{efficient}}] \label{lemma lower bounds}
For a set $ I \subseteq [n] $ and given preprocessing functions, the following bounds hold:
\begin{equation}
\ell \leq r - t, \label{eq lower bound info rate}
\end{equation}
\begin{equation}
{\rm CO}(I) \geq \frac{\ell t}{\#I - t}, \label{eq lower bound CO}
\end{equation}
\begin{equation}
{\rm DB}(I) \geq \frac{\ell \#I}{\#I - t}. \label{eq lower bound DB}
\end{equation}
\end{proposition}

Observe that the definition of communication overhead in \cite[Definition 2]{efficient} takes the values $ \log(\# \mathcal{B}_{I,i}) $ instead of the smaller values $ H(E(x_i)) $ {\color{black}(with base $ \# \mathcal{A} $)}. However, the proofs of the previous bounds work in the same way.

\subsection{Secret sharing schemes based on linear codes} \label{subsec SS based on nested codes}

A classical approach to constructing secret sharing schemes is by the so-called coset coding schemes. These were first considered by Wyner in \cite{wyner} for the problem of reliable and secure communication over wire-tap channels, which may be seen as secret sharing. A particular case is obtained when choosing the family of cosets as the quotient vector space of two linear codes over $ \mathbb{F}_q $, as considered in \cite{secure-computation, duursma, kurihara-secret, massey, mceliece, binning}, and which generalize Shamir's secret sharing scheme \cite{shamir}. The alphabet of these schemes is $ \mathcal{A} = \mathbb{F}_q $.

\begin{definition} [\textbf{Nested coset coding schemes \cite{secure-computation, duursma, kurihara-secret, massey, mceliece, binning}}] \label{definition nested coset}
A nested linear code pair is a pair of linear codes $ \mathcal{C}_2 \subsetneqq \mathcal{C}_1 \subseteq \mathbb{F}_q^n $. Choose a vector space $ \mathcal{W} $ such that $ \mathcal{C}_1 = \mathcal{C}_2 \oplus \mathcal{W} $ and a vector space isomorphism $ \psi : \mathbb{F}_q^\ell \longrightarrow \mathcal{W} $, where $ \ell = \dim(\mathcal{C}_1/\mathcal{C}_2) $. 

The nested coset coding scheme associated to $ \mathcal{C}_1 $ and $ \mathcal{C}_2 $ is the secret sharing scheme $ F : \mathbb{F}_q^\ell \longrightarrow \mathbb{F}_q^n $ defined by $ F(\mathbf{s}) = \psi(\mathbf{s}) + \mathbf{c} $, for $ \mathbf{s} \in \mathbb{F}_q^\ell $, where $ \mathbf{c} $ is the uniform random variable on $ \mathcal{C}_2 $.
\end{definition}

To evaluate the reconstruction and privacy thresholds $ r $ and $ t $ of a nested coset coding scheme, we need the concept of minimum Hamming distance of a nested linear code pair $ \mathcal{C}_2 \subsetneqq \mathcal{C}_1 \subseteq \mathbb{F}_q^n $ \cite{duursma} and of a linear code $ \mathcal{C} \subseteq \mathbb{F}_q^n $, defined respectively as follows:
$$ d(\mathcal{C}_1, \mathcal{C}_2) = \min \{ {\rm wt}(\mathbf{c}) \mid \mathbf{c} \in \mathcal{C}_1, \mathbf{c} \notin \mathcal{C}_2 \}, \textrm{ and} $$
$$ d(\mathcal{C}) = \min \{ {\rm wt}(\mathbf{c}) \mid \mathbf{c} \in \mathcal{C}_1, \mathbf{c} \neq \mathbf{0} \} = d(\mathcal{C}, \{ \mathbf{0} \}), $$
where $ {\rm wt}(\mathbf{c}) = \# \{ i \in [n] \mid c_i \neq 0 \} $ is the Hamming weight of the vector $ \mathbf{c} \in \mathbb{F}_q^n $.

The following result is proven in \cite[Corollary 1.7]{duursma}:

\begin{lemma} [\textbf{\cite{duursma}}]
The nested coset coding scheme based on the nested linear code pair $ \mathcal{C}_2 \subsetneqq \mathcal{C}_1 \subseteq \mathbb{F}_q^n $ is an $ (n, \ell, r, t)_{\mathbb{F}_q} $ secret sharing scheme with $ \ell = \dim(\mathcal{C}_1 / \mathcal{C}_2) $ and
\begin{equation}
r = n - d(\mathcal{C}_1, \mathcal{C}_2) + 1, \label{eq reconstruction charact}
\end{equation} 
\begin{equation}
t = d(\mathcal{C}_2^\perp, \mathcal{C}_1^\perp) + 1.  \label{eq privacy charact}
\end{equation}
\end{lemma}

\begin{remark} \label{remark degenerate dummies}
We will implicitly assume that $ \mathcal{C}_1 $ is not degenerate. That is, $ \mathcal{C}_{1 \{ i \}} \neq \{ 0 \} $, for every $ i \in [n] $. If $ \mathcal{C}_{1 \{ i \}} = \{ 0 \} $ for some $ i \in [n] $, then the $ i $-th share not only contains no information about the secret, but added to any set of shares, it does not increase the overall amount of information, hence it may be removed.
\end{remark}

\section{Communication efficient secret sharing schemes based on nested linear codes}  \label{sec constructions}

We will now present a general construction of communication efficient secret sharing schemes based on linear codes that is inspired by the constructions obtained independently in \cite{rawad} and \cite[Section IV]{efficient}. {\color{black}Our construction generalizes the previous two. We will explain how at the end of each subsection. }

We will present two versions of it. The first reduces communication overheads by contacting any $ \delta \geq r $ shares for one fixed value of $ \delta $, while the second reduces communication overheads by contacting any $ \delta \geq r $ shares for {\color{black}all $ \delta $ in an arbitrary set $ \Delta \subseteq [r,n] $}, hence achieving more flexibility.

In the rest of the paper, the alphabet is $ \mathcal{A} = \mathbb{F}_q^\alpha $, for a positive integer $ \alpha $. The advantage of the first construction with respect to the second is a significant smaller parameter $ \alpha $, hence a significant smaller alphabet.

We will measure information taking logarithms in base $ q $, {\color{black}instead of $ q^\alpha $. This only affects entropy and mutual information by product with $ \alpha $ (see \cite[Lemma 2.1.2]{cover}). }

\subsection{Construction 1: Non-universal but small alphabet} \label{subsec construction 1}

Take linear codes $ \mathcal{C}_2 \subsetneqq \mathcal{C}_1 \subseteq \mathcal{C} \subseteq \mathbb{F}_q^n $. {\color{black}The codes $ \mathcal{C}_1 $ and $ \mathcal{C}_2 $ will play the same role as in Subsection \ref{subsec SS based on nested codes}. The staircase structure of Construction 1 combined with the larger code $ \mathcal{C} $ makes it possible to recover the secret without downloading some (depending on the parameters of $ \mathcal{C} $ and $ \mathcal{C}_2 $) of the random messages introduced by $ \mathcal{C}_2 $ (see Definition \ref{definition nested coset}), hence reducing the decoding bandwidth. This idea was discussed in the original staircase constructions (see \cite[Example 2]{rawad}). }

We now {\color{black}fix positive integers as follows (inequalities are used to apply lower bounds on minimum Hamming distances):}
\begin{enumerate}
\item
$ \ell = k_1 - k_2 $, where $ k_1 = \dim(\mathcal{C}_1) $ and $ k_2 = \dim(\mathcal{C}_2) $,
\item
$ t \leq d(\mathcal{C}_2^\perp, \mathcal{C}^\perp) - 1 $ and $ r \geq n - d(\mathcal{C}_1, \mathcal{C}_2) + 1 $,
\item
$ \delta \geq n - d(\mathcal{C}, \mathcal{C}_2) + 1 $ such that $ \delta \geq r $,
\item
{\color{black}$ \alpha = k - k_2 \geq \ell $, where $ k = \dim(\mathcal{C}) $.}
\end{enumerate}
Finally, {\color{black}fix} a generator matrix of $ \mathcal{C} $ of the form:
\begin{displaymath}
G = \left(
\begin{array}{c}
G_1 \\
G_3
\end{array} \right) \in \mathbb{F}_q^{k \times n}, \quad \textrm{and} \quad G_1 = \left(
\begin{array}{c}
G_2 \\
G_c
\end{array} \right) \in \mathbb{F}_q^{k_1 \times n},
\end{displaymath} 
where $ G_2 \in \mathbb{F}_q^{k_2 \times n} $ is a generator matrix of $ \mathcal{C}_2 $ and $ G_1 \in \mathbb{F}_q^{k_1 \times n} $ is a generator matrix of $ \mathcal{C}_1 $.

The secret space is the space of matrices $ \mathcal{A}^\ell = \mathbb{F}_q^{\alpha \times \ell} $ (recall that $ \mathcal{A} = \mathbb{F}_q^\alpha $). Now, for a secret $ S \in \mathbb{F}_q^{\alpha \times \ell} $, divide it into two subsecrets as follows:
\begin{displaymath}
S = \left(
\begin{array}{c}
S_1 \\
S_2 
\end{array} \right), \quad S_1 \in \mathbb{F}_q^{\ell \times \ell} \textrm{ and } S_2 \in \mathbb{F}_q^{(\alpha - \ell) \times \ell},
\end{displaymath}
where $ S = S_1 $ in the case $ \alpha = \ell $. Next, generate uniformly at random a matrix
\begin{displaymath}
R = \left(
\begin{array}{c}
R_1 \\
R_2 
\end{array} \right) \in \mathbb{F}_q^{\alpha \times k_2}, \quad R_1 \in \mathbb{F}_q^{\ell \times k_2} \textrm{ and } R_2 \in \mathbb{F}_q^{(\alpha - \ell) \times k_2},
\end{displaymath} 
where $ R = R_1 $ in the case $ \alpha = \ell $. Finally, the $ i $-th share, where $ 1 \leq i \leq n $, is the $ i $-th column of the matrix
\begin{displaymath}
C = (\mathbf{c}_1, \mathbf{c}_2, \ldots, \mathbf{c}_n) = \left(
\begin{array}{c|c|c}
R_1 & S_1 & S_2^T \\ \hline
R_2 & S_2 & 0
\end{array} \right) G \in \mathbb{F}_q^{\alpha \times n},
\end{displaymath}
which is a column vector $ \mathbf{c}_i \in \mathbb{F}_q^{\alpha \times 1} $, that is, a symbol in our alphabet $ \mathcal{A} = \mathbb{F}_q^\alpha $.

For a set $ I \subseteq [n] $ and $ i \in I $, the corresponding preprocessing function is
\begin{equation*}
E : \mathbb{F}_q^\alpha \longrightarrow \mathbb{F}_q^{\ell},  
\end{equation*}
where $ E(\mathbf{c}_i) $ is obtained by restricting $ \mathbf{c}_i \in \mathbb{F}_q^{\alpha \times 1} $ to its first $ \ell $ rows. We will need the following lemma, which is a direct consequence of \cite[Lemma 1]{liu} together with \cite[Lemma 2]{forney}, as explained in the proof of \cite[Theorem 4]{kurihara-secret}. We recall the proof for convenience of the reader.

\begin{lemma} [\textbf{\cite{forney, kurihara-secret, liu}}] \label{lemma coset distance info}
If $ I \subseteq [n] $ is such that $ \# I < d(\mathcal{C}_2^\perp, \mathcal{C}^\perp) $, then $ \mathcal{C}_{2 I} = \mathcal{C}_I $.
\end{lemma}
\begin{proof}
{\color{black}Recall the notation and definitions of restricted and shortened codes from Subsection \ref{subsec notation}. The result} \cite[Lemma 1]{liu} states that
$$ d(\mathcal{C}_2^\perp, \mathcal{C}^\perp) = \min \{ \# I \mid \dim((\mathcal{C}_2^\perp)^I / (\mathcal{C}^\perp)^I) = 1 \}. $$
On the other hand, \cite[Lemma 2]{forney} implies that
$$ \dim((\mathcal{C}_2^\perp)^I / (\mathcal{C}^\perp)^I) = \dim(\mathcal{C}_I / \mathcal{C}_{2 I}), $$
hence the result follows.
\end{proof}

Now we may prove the main result of this subsection:

\begin{theorem} \label{theorem general staircase properties}
The previous secret sharing scheme has information rate $ \ell / n $, reconstruction $ r $, privacy $ t $ and, for any $ I \subseteq [n] $ with $ \# I = \delta $, it holds that
$$ {\rm CO}(I) = \frac{\ell (\delta - k + k_2)}{k - k_2}, \quad \textrm{or} \quad {\rm DB}(I) = \frac{\ell \delta}{k - k_2}. $$
\end{theorem}
\begin{proof}
We prove each item separately.

1) \textit{Reconstruction $ r $:} Take $ I \subseteq [n] $ of size at least $ r $. From
$$ (R_2 | S_2 | 0) G_I = (R_2 | S_2) G_{1 I} $$
we obtain $ S_2 $ by (\ref{eq reconstruction charact}), since $ \# I \geq n - d(\mathcal{C}_1, \mathcal{C}_2) + 1 $. On the other hand,
$$ (R_1 | S_1 | S_2^T) G_I = (R_1 | S_1)G_{1 I} + S_2^T G_{3 I}. $$
By substracting $ S_2^T G_{3 I} $, we see that we only need to decode
$$ (R_1 | S_1) G_{1 I}. $$
Again, we obtain $ S_1 $ by (\ref{eq reconstruction charact}), since $ \# I \geq n - d(\mathcal{C}_1, \mathcal{C}_2) + 1 $. Hence, we have obtained the whole secret $ S $ and proven that $ H(S | X_I) = 0 $.

2) \textit{Privacy $ t $:} Take $ I \subseteq [n] $ of size at most $ t $. The eavesdropper obtains
\begin{displaymath}
W_I = \left(
\begin{array}{c|c|c}
R_1 & S_1 & S_2^T \\ \hline
R_2 & S_2 & 0
\end{array} \right) G_I.
\end{displaymath}
This random variable $ W_I $ has support inside the linear code 
$$ \mathcal{C}_I^\alpha = \{ X G_I \mid X \in \mathbb{F}_q^{\alpha \times k} \} \subseteq \mathbb{F}_q^{\alpha \times \#I}. $$
Recall from \cite[Theorem 2.6.4]{cover} that, if a random variable $ X $ has support in the set $ \mathcal{X} $, then $ H(X) \leq \log(\# \mathcal{X}) $. Hence
$$ H(W_I) \leq \log(\# \mathcal{C}_I^\alpha) = \dim(\mathcal{C}_I^\alpha) = \alpha \dim(\mathcal{C}_I). $$
On the other hand, using the analogous notation $ \mathcal{C}_{2 I}^\alpha $ for $ \mathcal{C}_2 $ instead of $ \mathcal{C} $, it holds that
$$ H(W_I \mid S) = \log(\# \mathcal{C}_{2 I}^\alpha) = \dim(\mathcal{C}_{2 I}^\alpha) = \alpha \dim(\mathcal{C}_{2 I}), $$
since, given a value of $ S $, the variable $ W_I $ is a uniform random variable on an affine space obtained by translating the vector space $ \mathcal{C}_{2 I}^\alpha $. Hence we obtain that 
$$ I(S; W_I) = H(W_I) - H(W_I \mid S) $$
$$ \leq \alpha (\dim(\mathcal{C}_I) - \dim(\mathcal{C}_{2 I})) = 0, $$
where the last equality follows from Lemma \ref{lemma coset distance info} since $ \# I \leq d(\mathcal{C}_2^\perp, \mathcal{C}^\perp) -1 $.

3) \textit{Communication overhead and decoding bandwidth:} Take $ I \subseteq [n] $ of size $ \delta $. By definition, $ E(\mathbf{c}_i) $ is the $ i $-th column of the matrix
$$ (R_1 | S_1 | S_2^T) G \in \mathbb{F}_q^{\ell \times n}, $$
for each $ i \in I $. Thus we may obtain $ S_1 $ and $ S_2^T $ by (\ref{eq reconstruction charact}), since $ \# I \geq n - d(\mathcal{C}, \mathcal{C}_2) + 1 $, and hence we may obtain the whole secret $ S $. In other words, (\ref{eq condition preprocessing functions}) is satisfied. 

On the other hand, $ E(\mathbf{c}_i) $ is the uniform random variable on 
$$ \mathcal{C}_i^\ell = \{ X G_{ \{ i \} } \mid X \in \mathbb{F}_q^{\ell \times k} \}. $$ 
Since $ \mathcal{C} $ is not degenerate (see Remark \ref{remark degenerate dummies}), it follows that
$$ H_{q^\alpha}(E(\mathbf{c}_i)) = \frac{H(E(\mathbf{c}_i))}{\alpha} = \frac{\dim(\mathcal{C}_i^\ell)}{\alpha} = \frac{\ell}{\alpha}. $$
Hence, we have that
$$ {\rm DB}(I) = \frac{\ell \# I}{\alpha} = \frac{\ell \delta}{k - k_2}, \quad \textrm{and} $$
$$ {\rm CO}(I) = {\rm DB}(I) - \ell = \frac{\ell (\delta - k + k_2)}{k - k_2}. $$
\end{proof}

\begin{remark}
Observe that $ H(E(\mathbf{c}_i)) = \ell / \alpha $ does not depend on $ I $ nor $ i \in I $. Therefore, any choice of $ \delta $ available parties may be contacted with the same communication costs and moreover, the information transmission from these parties may be performed in a balanced and parallel way.
\end{remark}

{\color{black}We conclude by explaining how our Construction 1 generalizes that in \cite[Subsec. III-A]{rawad}. The construction in \cite[Section IV]{efficient} is equivalent. Choose} $ \mathcal{C}_2 \subsetneqq \mathcal{C}_1 \subseteq \mathcal{C} $ as nested Reed-Solomon codes of dimensions $ k_2 $, $ k_1 $ and $ k $, respectively. {\color{black}Choose} $ t = k_2 $, $ r = k_1 $ and $ \delta = k $. {\color{black}Observe that} the communication overhead for a set $ I \subseteq [n] $ of size $ \delta $ in the previous theorem coincides then with that in \cite[Equation (4)]{rawad}, which also coincides with the optimal value given by (\ref{eq lower bound CO}).
 
{\color{black}As is well-known, in this case we may take $ G_2 \in \mathbb{F}_q^{k_2 \times n} $, $ G_1 \in \mathbb{F}_q^{k_1 \times n} $ and $ G \in \mathbb{F}_q^{k \times n} $ as nested Vandermonde matrices. The shares in our Construction 1 are given by the columns in
\begin{displaymath}
 \left(
\begin{array}{c|c|c}
R_1 & S_1 & S_2^T \\ \hline
R_2 & S_2 & 0
\end{array} \right) G \in \mathbb{F}_q^{\alpha \times n},
\end{displaymath}
whereas the shares in \cite[Subsec. III-A]{rawad} are given by the columns in
\begin{displaymath}
 \left(
\begin{array}{c|cc}
S_1^T & S_2^T & R_1 \\ \hline
D^T & R_2 & 0
\end{array} \right) G \in \mathbb{F}_q^{\alpha \times n},
\end{displaymath}
where $ D \in \mathbb{F}_q^{\ell \times (\alpha - \ell)} $ is given by the last $ \alpha - \ell $ columns in $ (S^T \mid R_1) \in \mathbb{F}_q^{\ell \times k} $.

We can see the equivalence of these two schemes as follows: First, considering $ S_1 $ or $ S_1^T $ makes no difference, since both are $ \ell \times \ell $ matrices and contain the same part of the secret matrix $ S $. Second, the placement of the matrices $ R_1 $ and $ R_2 $ makes no difference, since any choice of consecutive rows of a Vandermonde matrix can be seen again as a Vandermonde matrix. For the same reason, the matrix $ D $ could be chosen as any consecutive $ \alpha - \ell $ columns in $ (S^T \mid R_1) $, in particular, we may choose $ D = S_2^T $ as in our case.

These ways of restructuring the staircase construction work for Reed-Solomon codes due to the fact that any choice of consecutive rows of a Vandermonde matrices essentially gives a new Reed-Solomon code. This is not the case for general linear block codes. Hence it seems that our version of the staircase construction is necessary in general. }

\subsection{Construction 2: Universal but large alphabet} \label{subsec Construction 2}

In this subsection, we extend the universal constructions from {\color{black}\cite[Section VI]{rawad}} and \cite[Section IV]{efficient} to more general schemes and evaluate their communication overhead and decoding bandwidth. The objective is to obtain schemes with low communication overhead for all values of $ \delta $ {\color{black}in a subset $ \Delta \subseteq [r,n] $}, and not just one fixed value. The main requisites for this construction are a larger alphabet than in Construction 1 and the existence of an appropriate sequence of nested linear codes containing the main nested linear code pair{\color{black}, both depending on $ \Delta $.}

Take linear codes $ \mathcal{C}_2 \subsetneqq \mathcal{C}_1 \subseteq \mathbb{F}_q^n $ {\color{black}and fix positive integers:
\begin{enumerate}
\item
$ \ell = k_1 - k_2 $, where $ k_1 = \dim(\mathcal{C}_1) $ and $ k_2 = \dim(\mathcal{C}_2) $,
\item
$ r \geq n - d(\mathcal{C}_1, \mathcal{C}_2) + 1 $.
\end{enumerate} 
Fix now an arbitrary subset $ \Delta = \{ \delta_1, \delta_2, \ldots, \delta_h \} \subseteq [r, n] $, where we assume that $ r = \delta_h < \delta_{h-1} < \ldots < \delta_2 < \delta_1 \leq n $. Assume then that there exists a sequence of nested linear codes
$$ \mathcal{C}_1 = \mathcal{C}^{(h)} \subsetneqq \mathcal{C}^{(h-1)} \subsetneqq \ldots \subsetneqq \mathcal{C}^{(2)} \subsetneqq \mathcal{C}^{(1)} \subseteq \mathbb{F}_q^n, $$
such that we may consider:
\begin{enumerate}
\item[3)]
$ t \leq d(\mathcal{C}_2^\perp, \mathcal{C}^{(1) \perp}) -1 $ and $ \delta_j \geq n - d(\mathcal{C}^{(j)}, \mathcal{C}_2) + 1 $, 
\item[4)]
$ \alpha_j = k^{(j)} - k_2 $, where $ k^{(j)} = \dim(\mathcal{C}^{(j)}) $, 
\end{enumerate}
for $ j = 1,2, \ldots, h $. Finally define
$$ \alpha = {\rm LCM}(\alpha_1, \alpha_2, \ldots, \alpha_h). $$
The intuition behind the choice of the larger codes $ \mathcal{C}^{(j)} $ is the same as that of $ \mathcal{C} $ in Construction 1 (see the beginning of Subsection \ref{subsec construction 1}). Again, we define our parameters by inequalities in order to apply lower bounds on minimum Hamming distances. 

Take a generator matrix $ G_2 \in \mathbb{F}_q^{k_2 \times n} $ of $ \mathcal{C}_2 $ and a generator matrix $ G_1 \in \mathbb{F}_q^{k_1 \times n} $ of $ \mathcal{C}_1 $ of the form
\begin{displaymath}
G_1 = \left(
\begin{array}{c}
G_2 \\
G_c
\end{array} \right) \in \mathbb{F}_q^{k_1 \times n}.
\end{displaymath}
Recursively and decreasingly in $ j = h-1,h-2, \ldots, 2,1 $, take a generator matrix $ G^{(j)} \in \mathbb{F}_q^{k^{(j)} \times n} $ of $ \mathcal{C}^{(j)} $ of the form
\begin{displaymath}
G^{(j)} = \left(
\begin{array}{c}
G^{(j+1)} \\
G^{(j+1)}_c
\end{array} \right) \in \mathbb{F}_q^{k^{(j)} \times n},
\end{displaymath} 
where $ G^{(h)} = G_1 $. Next define the following positive integers, which are analogous to those in \cite[Equation (11)]{efficient}:
\begin{displaymath}
p_j = \left\lbrace \begin{array}{ll}
\frac{\ell \alpha}{\alpha_1} & \textrm{if } j = 1, \\
\frac{\ell \alpha}{\alpha_j} - \frac{\ell \alpha}{\alpha_{j-1}} & \textrm{if } 1 < j \leq h.
\end{array} \right.
\end{displaymath}
}

The secret space is again $ \mathcal{A}^\ell = \mathbb{F}_q^{\alpha \times \ell} $ with alphabet $ \mathcal{A} = \mathbb{F}_q^\alpha $, and we also generate uniformly at random a matrix $ R \in \mathbb{F}_q^{\alpha \times k_2} $. In this case, we divide $ S $ and $ R $ as follows:
\begin{displaymath}
S = \left(
\begin{array}{c}
S_1 \\
S_2 \\
\vdots \\
S_h 
\end{array} \right), \quad R = \left(
\begin{array}{c}
R_1 \\
R_2 \\
\vdots \\
R_h 
\end{array} \right),
\end{displaymath}
where
$$ S_j \in \mathbb{F}_q^{ {\color{black} p_j } \times \ell} \textrm{ and } R_j \in \mathbb{F}_q^{{\color{black} p_j } \times k_2}, $$
for $ j = 1,2, \ldots, h $. Next, we define the matrices
\begin{displaymath}
\begin{array}{ccccccccc}
M_1 & = & (R_1 & S_1 & D_{1,1} & D_{1,2} & \ldots & D_{1,h-1} ), \\
M_2 & = & (R_2 & S_2 & D_{2,1} & D_{2,2} & \ldots & 0 ), \\
M_3 & = & (R_3 & S_3 & D_{3,1} & D_{3,2} & \ldots & 0 ), \\
\vdots & & & \vdots & & & & \\
M_{h-1} & = & (R_{h-1} & S_{h-1} & D_{h-1,1} & 0 & \ldots & 0 ), \\
M_h & = & (R_h & S_h & 0 & 0 & \ldots & 0 ), 
\end{array} 
\end{displaymath}
where $ M_u \in \mathbb{F}_q^{{\color{black} p_u } \times k^{(1)}} $, and where the {\color{black}matrices $ D_{u, v} \in \mathbb{F}_q^{ p_u \times (\alpha_{h-v} - \alpha_{h-v+1})} $} are defined column-wise iteratively as follows: For $ v = 1,2, \ldots, h-1 $, the components of the $ v $-th column {\color{black}block}
\begin{displaymath}
\left( \begin{array}{c}
D_{1, v} \\
D_{2, v} \\
\vdots \\
D_{h-v, v}
\end{array} \right) \in \mathbb{F}_q^{\ell \alpha / \alpha_{h-v} \times {\color{black}(\alpha_{h-v} - \alpha_{h-v+1})}},
\end{displaymath}
are the components (after some fixed rearrangement) of the matrix
$$ (S_{h-v+1} | D_{h-v+1, 1} | D_{h-v+1, 2} | \ldots | D_{h-v+1, v-1}), $$
whose size is {\color{black}$ p_{h-v+1} \times \alpha_{h-v+1} $. Observe that $ p_{j+1} \alpha_{j+1} = (\alpha_j - \alpha_{j+1}) \ell \alpha / \alpha_j $.} For convenience, we define the matrices
\begin{equation}
M^\prime_j = \left(
\begin{array}{c}
M_1 \\
M_2 \\
\vdots \\
M_j  
\end{array} \right) \in \mathbb{F}_q^{\ell \alpha / \alpha_j \times k^{(1)}}, \label{eq definition M_h}
\end{equation}
for $ j = 1,2, \ldots, h $. Observe also that $ M^\prime_h \in \mathbb{F}_q^{\alpha \times k^{(1)}} $ since $ \ell \alpha / \alpha_h = \alpha $. Finally, the $ i $-th share, where $ 1 \leq i \leq n $, is the $ i $-th column of the matrix 
$$ C = (\mathbf{c}_1, \mathbf{c}_2, \ldots, \mathbf{c}_n) = M^\prime_h G^{(1)} \in \mathbb{F}_q^{\alpha \times n}, $$
which is a column vector $ \mathbf{c}_i \in \mathbb{F}_q^{\alpha \times 1} $, that is, a symbol in our alphabet $ \mathcal{A} = \mathbb{F}_q^\alpha $. 

For $ j = 1,2, \ldots, h $, a set $ I \subseteq [n] $ of size $ \delta_j $ and $ i \in I $, the corresponding preprocessing function is 
$$ E_{j} : \mathbb{F}_q^\alpha \longrightarrow \mathbb{F}_q^{\ell \alpha / \alpha_j}, $$
where $ E_{j}(\mathbf{c}_i) $ is obtained by restricting $ \mathbf{c}_i \in \mathbb{F}_q^\alpha $ to its first $ \ell \alpha / \alpha_j $ rows.

We may now establish the main result of this subsection:

\begin{theorem} \label{theorem general univeral staircase properties}
The previous secret sharing scheme has information rate $ \ell / n $, reconstruction $ r $, privacy $ t $ and, for any {\color{black}$ \delta_j \in \Delta $} and any $ I \subseteq [n] $ with $ \# I = \delta_j $, it holds that
$$ {\rm CO}(I) = \frac{\ell (\delta_j - k^{(j)} + k_2)}{k^{(j)} - k_2}, \quad \textrm{or} \quad {\rm DB}(I) = \frac{\ell \delta_j}{k^{(j)} - k_2}. $$
\end{theorem}
\begin{proof}
We proceed as in the proof of Theorem \ref{theorem general staircase properties}:

1) \textit{Reconstruction $ r $:} Take $ I \subseteq [n] $ of size at least $ r $. From
$$ (R_h | S_h | 0) G_I = (R_h | S_h) G_{1 I} $$
we obtain $ S_h $ by (\ref{eq reconstruction charact}), since $ \# I \geq n - d(\mathcal{C}_1, \mathcal{C}_2) + 1 $. By definition, we have obtained $ D_{u,1} $, for $ u = 1,2, \ldots, h-1 $. Hence, substracting $ D_{h-1,1} (G^{(h-1)}_c)_I $ from $ (R_{h-1} | S_{h-1} | D_{h-1,1} | $ $ 0)G^{(1)}_I $ as in the proof of Theorem \ref{theorem general staircase properties}, we may obtain 
$$ (R_{h-1} | S_{h-1} ) G_{1 I}, $$
and thus we obtain $ S_{h-1} $ by (\ref{eq reconstruction charact}), since $ \# I \geq n - d(\mathcal{C}_1, \mathcal{C}_2) + 1 $. Now, we have also obtained $ D_{u,2} $, for $ u = 1,2, \ldots, h-2 $. Proceeding iteratively in the same way, we see that we may obtain all the matrices $ S_j $, for $ j = 1,2, \ldots, h $, and thus we obtain the whole secret $ S $. In particular, we have shown that $ H(S | X_I) = 0 $.

2) \textit{Privacy $ t $:} Take $ I \subseteq [n] $ of size at most $ t $, and assume that the eavesdropper obtains
$$ W_I = M^\prime_h G^{(1)}_I. $$
As in the proof of Theorem \ref{theorem general staircase properties}, we have that
$$ I(S; W_I) = H(W_I) - H(W_I \mid S) $$
$$ \leq \alpha (\dim(\mathcal{C}^{(1)}_I) - \dim(\mathcal{C}_{2 I})) = 0, $$
where the last equality follows from Lemma \ref{lemma coset distance info}, since $ \# I \leq d(\mathcal{C}_2^\perp, \mathcal{C}^{(1) \perp}) - 1 $.

3) \textit{Communication overhead and decoding bandwidth:} Take $ I \subseteq [n] $ of size $ \delta_j $ for some $ 1 \leq j \leq h $. By definition, $ E_{j}(\mathbf{c}_i) $ is the $ i $-th column of the matrix
$$ M_j^\prime G^{(1)} \in \mathbb{F}_q^{\ell \alpha / \alpha_j \times n}, $$
for each $ i \in I $. As in the proof of Theorem \ref{theorem general staircase properties}, we may obtain the matrix $ M_j^\prime $ by (\ref{eq reconstruction charact}), since $ \# I \geq n - d(\mathcal{C}^{(j)}, \mathcal{C}_2) + 1 $. By definition, the matrices $ S_1, S_2, \ldots, S_j $ are contained in $ M_j^\prime $. On the other hand, the {\color{black}matrices $ D_{1,h-j}, D_{2,h-j}, \ldots, D_{j,h-j} $} are also contained in $ M_j^\prime $, and from them we obtain by definition $ S_{j+1} $ and $ D_{j+1,1}, D_{j+1,2}, \ldots, D_{j+1,h-j-1} $. Now, the {\color{black}matrices $ D_{1,h-j-1}, D_{2,h-j-1}, \ldots, D_{j,h-j-1} $} are contained in $ M_j^\prime $ and we also have $ D_{j+1,h-j-1} $, hence we may obtain by definition $ S_{j+2} $ and $ D_{j+2,1}, D_{j+2,2}, \ldots, D_{j+2,h-j-2} $. 

Continuing iteratively in this way, we may obtain all $ S_1, S_2, \ldots, S_h $ and hence the secret matrix $ S $. In other words, we have proven that the preprocessing functions satisfy (\ref{eq condition preprocessing functions}).

Finally, $ E_{j}(\mathbf{c}_i) $ is the uniform random variable on 
$$ (\mathcal{C}_i^{(j)})^{\ell \alpha / \alpha_j} = \{ X G_{ \{ i \} }^{(j)} \mid X \in \mathbb{F}_q^{\ell \alpha / \alpha_j  \times k^{(j)}} \}. $$ 
Since $ \mathcal{C}^{(j)} $ is not degenerate (see Remark \ref{remark degenerate dummies}), it follows that
$$ H_{q^\alpha}(E_{j}(\mathbf{c}_i)) = \frac{H(E_{j}(\mathbf{c}_i))}{\alpha} $$
$$ = \frac{\dim((\mathcal{C}_i^{(j)})^{\ell \alpha / \alpha_j})}{\alpha} = \frac{\ell \alpha}{\alpha_j} \cdot \frac{1}{\alpha} = \frac{\ell}{\alpha_j}. $$
Hence, we have that
$$ {\rm DB}(I) = \frac{\ell \# I}{\alpha_j} = \frac{\ell \delta_j}{k^{(j)} - k_2}, \quad \textrm{and} $$
$$ {\rm CO}(I) = {\rm DB}(I) - \ell = \frac{\ell (\delta_j - k^{(j)} + k_2)}{k^{(j)} - k_2}. $$
\end{proof}

\begin{remark}
Observe, as in the previous subsection, that $ H(E_{j}(\mathbf{c}_i)) = \ell / \alpha_j $ depends on $ j $ (that is, on $ \delta=\delta_j $) but does not depend on $ I $ nor $ i \in I $. Therefore, for a fixed {\color{black}$ \delta \in \Delta $}, any $ \delta $ available parties may be contacted with the same communication costs, and the information transmission from these parties may be performed in a balanced and parallel way.
\end{remark}

As in the previous subsection, our Construction 2 specializes to the construction in {\color{black}\cite[Section VI]{rawad}} and \cite[Section IV]{efficient} by choosing $ \mathcal{C}_2 \subsetneqq \mathcal{C}_1 $ and $ \mathcal{C}^{(j)} $ as nested Reed-Solomon codes of dimensions $ k_2 $, $ k_1 $ and $ k^{(j)} $, respectively, and where $ t = k_2 $, $ r = k_1 $ and $ \delta_j = k^{(j)} $, for $ j = 1,2, \ldots, h $. The communication overhead for a set $ I \subseteq [n] $ of size $ \delta_j $ in the previous theorem coincides then with that in \cite[Equation (4)]{rawad}, which also coincides with the optimal value given by (\ref{eq lower bound CO}). {\color{black}The details of how our Construction 2 extends that in \cite[Section IV]{rawad} is analogous to those given at the end of Subsection \ref{subsec construction 1}, although the notation is now more cumbersome.}

\section{Communication efficient secret sharing schemes based on algebraic geometry codes}  \label{sec AG codes}

In this section, we will see that algebraic geometry codes fit into the previous two constructions of communication efficient secret sharing schemes, and allow to obtain schemes with arbitrarily large length $ n $ while keeping the field size $ q $ fixed and keeping low communication overhead, although not necessarily optimal. When using appropriate sequences of algebraic geometry codes, we will see that the difference with the optimal case only depends on $ q $, and not on the length $ n $. Hence, for a given defect with respect to the optimal case, we only need to fix a suitable field size $ q $, and then let $ n $ be arbitrarily large.

Algebraic geometry codes fit in both Construction 1 and Construction 2 since it is easy to find sequences of nested algebraic geometry codes contained in and containing a given one, as we will see. We refer to \cite{handbook}, \cite[Section 12.7]{cramerbook} and \cite{stichtenothbook} for general references on algebraic geometry codes. {\color{black}We also remark here that generator matrices of algebraic geometry codes (thus secret sharing schemes) can be explicitly and efficiently constructed in many cases (for instance, the codes in \cite{notehermitian, tiersma, vanlint-springer} or \cite{garciatower, shum}). A good brief explanation of what is needed is given at the end of \cite[Section 2.3]{stichtenothbook}. Basically, one needs to obtain expressions of rational places and Riemann-Roch spaces (see also Example \ref{example generator matrix Hermitian}), which we define now.}

Consider an irreducible projective curve $ \mathcal{X} $ over $ \mathbb{F}_q $ (which in this paper means irreducible over the algebraic closure of $ \mathbb{F}_q $) with algebraic function field $ \mathcal{F} $, and let $ g = g(\mathcal{X}) = g(\mathcal{F}) $ be its genus. Points in $ \mathcal{X} $ correspond with places in $ \mathcal{F} $ and we say that they are rational if they are rational over $ \mathbb{F}_q $ (have coordinates over $ \mathbb{F}_q $). A divisor over $ \mathcal{X} $ is a formal sum $ D = \sum_{P \in \mathcal{X}} \mu_P P $, for integers $ \mu_P \in \mathbb{Z} $ which are all zero except for a finite number. The support of $ D $ is defined as $ \{ P \in \mathcal{X} \mid \mu_P \neq 0 \} $, and $ D $ is called rational if all points in its support are rational. We define the degree of the rational divisor $ D $ as $ \deg(D) = \sum_{P \in \mathcal{X}} \mu_P \in \mathbb{Z} $. All divisors considered in this paper {\color{black}(except in Appendix \ref{app multiplicative})} will be rational.

On the other hand, for divisors $ D = \sum_{P \in \mathcal{X}} \mu_P P $ and $ E = \sum_{P \in \mathcal{X}} \lambda_P P $, we write $ D \preceq E $ if $ \mu_P \leq \lambda_P $, for all $ P \in \mathcal{X} $. For an algebraic function $ f \in \mathcal{F} $, we define its divisor as $ (f) = \sum_{P \in \mathcal{X}} \nu_P(f) P $, where $ \nu_P $ is the valuation at the point $ P $ (see \cite[Definition 1.1.12]{stichtenothbook} or \cite[Definition 2.15]{handbook}). Hence we may define the Riemann-Roch space (see \cite[Definition 1.1.4]{stichtenothbook} or \cite[Definition 2.36]{handbook}) of a divisor $ D $ as the vector space over $ \mathbb{F}_q $ given by:
\begin{equation*}
\mathcal{L}(D) = \{ f \in \mathcal{F} \mid (f) + D \succeq 0 \}.
\end{equation*}
Finally, for rational divisors $ D = P_1 + P_2 + \cdots + P_n $ and $ G $ over $ \mathcal{X} $ with disjoint supports and where the points $ P_i $ are pairwise distinct, we define the corresponding algebraic geometry code (see \cite[Equation (2.3)]{stichtenothbook} or \cite[Definition 2.64]{handbook}), or AG code for short, as the following linear code:
\begin{equation*}
\mathcal{C}(D,G) = \{ (f(P_1), f(P_2), \ldots, f(P_n)) \mid f \in \mathcal{L}(G) \} \subseteq \mathbb{F}_q^n.
\end{equation*}

We will need the following two well-known results on the parameters of algebraic geometry codes. The first is the well-known Goppa bound \cite[Corollary 2.2.3 (a)]{stichtenothbook} (see also \cite[Theorem 2.65]{handbook}), together with its dual statement (see \cite[Theorem 2.2.7]{stichtenothbook} and \cite[Theorem 2.2.8]{stichtenothbook}):

\begin{lemma} [\textbf{Goppa bound \cite{handbook,stichtenothbook}}] \label{lemma goppa}
If $ \deg(G) < n $ and $ \mathcal{C} = \mathcal{C}(D,G) $, or if $ 2g - 2 < \deg(G) < n $ and $ \mathcal{C} = \mathcal{C}(D,G)^\perp $, then 
\begin{equation*}
d(\mathcal{C}) \geq n - \dim(\mathcal{C}) - g + 1. 
\end{equation*} 
\end{lemma} 

The following lemma is \cite[Corollary 2.2.3(b)]{stichtenothbook} (see also \cite[Theorem 2.65]{handbook}):

\begin{lemma}[\textbf{\cite{handbook,stichtenothbook}}] \label{lemma dimensions AG}
If $ 2g-2 < \deg(G) < n $, then
$$ \dim(\mathcal{C}(D,G)) = \deg(G) - g + 1. $$
\end{lemma}

\subsection{Algebraic geometry codes for {\color{black}Constructions 1 and 2}} \label{subsec Cons 1 and 2 for AG codes}

Let the notation be as in the beginning of this section. In the following {\color{black}two propositions}, we gather the parameters obtained in {\color{black}Constructions 1 and 2} when using AG codes. {\color{black}We start presenting the result for Construction 1. We omit its proof, since it is analogous and simpler to that of Construction 2.

\begin{proposition} \label{proposition using AG for construction 1}
Let $ \mathcal{X} $ be an irreducible projective curve over $ \mathbb{F}_q $, let $ P_1, P_2, \ldots, P_n $ be pairwise distinct rational points on $ \mathcal{X} $, and let $ G_2 $, $ G_1 $ and $ G $ be rational divisors on $ \mathcal{X} $ whose supports do not contain the previous points $ P_1, P_2, \ldots, P_n $, and such that 
$$ G_2 \preceq G_1 \preceq G, \quad \textrm{and} $$ 
$$ 2g-2 < \deg(G_2) < \deg(G_1) \leq \deg(G) < n. $$

The secret sharing scheme in Construction 1 using the linear codes
\begin{equation*}
\mathcal{C}_2 = \mathcal{C}(D, G_2), \quad \mathcal{C}_1 = \mathcal{C}(D, G_1), \quad \mathcal{C} = \mathcal{C}(D, G),
\end{equation*}
where $ D = P_1 + P_2 + \cdots + P_n $, has the following parameters:
\begin{enumerate}
\item
Information rate $ \ell / n = (\deg(G_1) - \deg(G_2)) / n $, 
\item
Reconstruction $ r = \deg(G_1) + 1 $,
\item
Privacy $ t = \deg(G_2) - 2g + 1 $,
\end{enumerate}
and for any $ I \subseteq [n] $ of size $ \delta = \deg(G) + 1 $, it holds that
$$ {\rm CO}(I) = \frac{\ell (t + 2g)}{\delta - t - 2g}, \quad \textrm{or} \quad {\rm DB}(I) = \frac{\ell \delta}{\delta - t - 2g}. $$
\end{proposition}

We now present and prove the analogous result for Construction 2. For simplicity, we gather the main assumptions in the following paragraph:

Assumptions (A):} Let $ \mathcal{X} $ be an irreducible projective curve over $ \mathbb{F}_q $ and let $ P_1, P_2, \ldots, P_n $ be pairwise distinct rational points on $ \mathcal{X} $. Let $ 2g < r \leq n $, {\color{black}let $ \Delta \subseteq [r,n] $ with elements $ r = \delta_h < \delta_{h-1} < \ldots < \delta_2 < \delta_1 \leq n $,} and let $ G_2, G_1 = G^{(h)}, G^{(h-1)}, \ldots, G^{(2)}, G^{(1)} $ be rational divisors on $ \mathcal{X} $ whose supports do not contain the previous points $ P_1, P_2, \ldots, P_n $. Define $ D = P_1 + P_2 + \cdots + P_n $ and assume also that 
$$ G_2 \preceq G_1 = G^{(h)} \preceq G^{(h-1)} \preceq \ldots \preceq G^{(2)} \preceq G^{(1)}, $$ 
$$ 2g-2 < \deg(G_2) < \deg(G_1) = r- 1 < n $$
and $ \deg(G^{(j)}) = \delta_j - 1 $, for $ j = 1, 2, \ldots, h $.  

\begin{proposition} \label{proposition using AG for construction 2}
{\color{black}If the assumptions (A) hold, then} the secret sharing scheme in Construction 2 using the linear codes
\begin{equation*}
\mathcal{C}_2 = \mathcal{C}(D, G_2), \quad \mathcal{C}_1 = \mathcal{C}(D, G_1), \quad \mathcal{C}^{(j)} = \mathcal{C}(D, G^{(j)}),
\end{equation*}
{\color{black}for $ j = 1,2, \ldots, h $, has the following parameters:
\begin{enumerate}
\item
Information rate $ \ell / n = (\deg(G_1) - \deg(G_2)) / n $, 
\item
Reconstruction $ r = \deg(G_1) + 1 $,
\item
Privacy $ t = \deg(G_2) - 2g + 1 $,
\end{enumerate}
and, for any $ \delta \in \Delta $} and any $ I \subseteq [n] $ of size $ \delta $, it holds that {\color{black}
$$ {\rm CO}(I) = \frac{\ell (t + 2g)}{\delta - t - 2g}, \quad \textrm{or} \quad {\rm DB}(I) = \frac{\ell \delta}{\delta - t - 2g}. $$}
\end{proposition}
\begin{proof}
It follows from combining Theorem \ref{theorem general univeral staircase properties}, Lemma \ref{lemma goppa} and Lemma \ref{lemma dimensions AG}. Observe that, since $ G_2 \preceq G_1 $ and $ G^{(j)} \preceq G^{(j-1)} $, then $ \mathcal{L}(G_2) \subseteq \mathcal{L}(G_1) $ and $ \mathcal{L}(G^{(j)}) \subseteq \mathcal{L}(G^{(j-1)}) $ and hence $ \mathcal{C}_2 \subseteq \mathcal{C}_1 $ and $ \mathcal{C}^{(j)} \subseteq \mathcal{C}^{(j-1)} $, for $ j = 2,3, \ldots, h $, {\color{black}where inclusions are strict due to Lemma \ref{lemma dimensions AG}, since $ 2g - 2 < \deg(G_2) < \deg(G_1) < \deg(G^{(j)}) < n $, for $ j = 1,2, \ldots, h-1 $.} 
\end{proof}

Observe that, if $ r \leq \delta < \delta^\prime \leq n $, then
$$ \frac{\ell (t + 2g)}{\delta - t - 2g} > \frac{\ell (t + 2g)}{\delta^\prime - t - 2g}, $$
and therefore the communication overhead decreases as the number of contacted shares $ \delta $ increases, as expected.

As remarked in Section \ref{sec intro}, we can see that the defect with respect to the optimal case is only on the privacy thresholds, meaning that if $ \ell $ and $ r $ are given, $ t = r - \ell - 2g $ differs from the optimal value $ r-\ell $ by $ 2g $, but the decoding bandwidths present the same improvement on $ r $ as the optimal constructions in \cite{rawad, efficient} {\color{black}(see also Example \ref{example hermitian parameters})}:
\begin{equation}
{\rm DB}(I) = \frac{\ell \delta}{\ell + ( \delta -r)},
\label{eq optimality on r}
\end{equation}
for any {\color{black}$ \delta \in \Delta $} and any set $ I $ of size $ \delta $. Observe that, if $ t = r - \ell $ is optimal, then the optimal value given by (\ref{eq lower bound DB}) coincides with that in (\ref{eq optimality on r}) when substituting $ t $ by $ r-\ell $.

For more concrete constructions and existential proofs and to obtain the largest possible lengths, it is usual to consider the so-called one-point algebraic geometry codes. These are AG codes where the support of the divisor $ G $ is constituted by a single point. That is, $ G = \mu Q $ for some rational point $ Q \in \mathcal{X} $ and integer $ \mu \in \mathbb{Z} $. It holds that $ \deg(G) = \mu $. The next consequence follows:

\begin{corollary} \label{corollary SSS one-point construction2} 
If there exists an irreducible projective curve $ \mathcal{X} $ over $ \mathbb{F}_q $ with $ N \geq n+1 $ rational points and genus $ g $, then for any $ 0 \leq t < r \leq n $ with $ r - t > 2g $ {\color{black}and any $ \Delta \subseteq [r,n] $ with $ r \in \Delta $}, there exists a secret sharing scheme with information rate $ \ell / n = (r - t - 2g)/n $, reconstruction $ r $, privacy $ t $ and, for any {\color{black}$ \delta \in \Delta $} and any $ I \subseteq [n] $ of size $ \delta $, it holds that
$$ {\rm CO}(I) = \frac{\ell (t + 2g)}{\delta - t - 2g}, \quad \textrm{or} \quad {\rm DB}(I) = \frac{\ell \delta}{\delta - t - 2g}. $$
{\color{black}Moreover, the scheme can be explicitly constructed if expressions of the rational places and Riemann-Roch spaces corresponding to $ \mathcal{X} $ are known, due to the description in Proposition \ref{proposition using AG for construction 2}. }
\end{corollary}
\begin{proof}
Take pairwise distinct rational points $ Q, P_1, $ $ P_2, $ $ \ldots, $ $ P_n \in \mathcal{X} $, and define $ \mu_2 = t + 2g - 1 $, $ \mu_1 = r - 1 $ and $ \mu^{(j)} = \delta_j - 1 $, for $ j = 1,2, \ldots,h $, {\color{black}and where the elements in $ \Delta $ are $ r = \delta_h < \delta_{h-1} < \ldots < \delta_2 < \delta_1 \leq n $}. Observe that $ t \geq 0 $ implies that $ \mu_2 > 2g-2 $, $ r > t + 2g $ implies that $ \mu_1 > \mu_2 $, and $ \delta_1 \leq n $ implies that $ \mu^{(1)} < n $. Defining $ G_2 = \mu_2 Q $, $ G_1 = \mu_1 Q $ and $ G^{(j)} = \mu^{(j)} Q $, for $ j = 1,2, \ldots, h $, the assumptions in the previous proposition are satisfied and thus the result follows. {\color{black}Finally, as remarked at the beginning of this section, the requirements to explicitly construct the corresponding schemes are explained at the end of \cite[Section 2.3]{stichtenothbook}.}
\end{proof}

\begin{remark} \label{remark near optimal low g}
Observe that the parameters in the previous corollary are close to the optimal values, in view of Lemma \ref{lemma lower bounds}, whenever the genus $ g $ is close to $ 0 $. As a particular case, taking curves of genus $ g = 0 $ (such as the projective line, which is the case of Reed-Solomon codes), we obtain schemes with optimal parameters, which are the ones obtained in \cite{rawad} and \cite{efficient}.

In this way, we have reduced the problem of finding long secret sharing schemes with low decoding bandwidth to the well studied problem of finding irreducible projective curves with many rational points and low genus \cite{ihara, stichtenothbook, vladut}.
\end{remark}

In the next subsection, we will give a well-known explicit family of algebraic geometry codes and apply it to the previous corollary.

\subsection{Using Hermitian codes}

In this subsection, we apply Corollary \ref{corollary SSS one-point construction2} to the case where $ \mathcal{X} $ is a Hermitian curve. These curves give the so-called Hermitian codes (see \cite{notehermitian, tiersma}, \cite[Section VI]{vanlint-springer} or \cite[Section VI]{one-point}), which are not MDS but are considered in many cases in practice to be better than Reed-Solomon codes for the same information rate, since they are considerably longer and still have a large minimum Hamming distance (See \cite[Section VI]{vanlint-springer} for a discussion).

Assume that the field size is a perfect square $ q = u^2 $, and consider the projective plane curve $ \mathcal{X} $ defined by $ x^{u+1} - y^uz - yz^u = 0 $. This curve is called a Hermitian curve and has $ N = u^3+1 = \sqrt{q}^3 + 1 $ rational points, and genus 
$$ g = g(\mathcal{X}) = \frac{u(u-1)}{2} = \frac{q - \sqrt{q}}{2}. $$

Specializing Corollary \ref{corollary SSS one-point construction2} to these curves, we obtain the following result:

\begin{proposition} \label{proposition hermitian construction2} 
Let $ \mathcal{X} $ be a Hermitian curve over $ \mathbb{F}_q $, where $ q = u^2 $. Let $ Q, P_1, P_2, \ldots, P_{u^3} $ be its $ N = u^3+1 $ rational points, {\color{black}let $ u^2 - u + 2 \leq n \leq u^3 $ and define $ D = P_1 + P_2 + \cdots + P_n $}. Let $ 0 \leq t < r \leq n $ be such that $ r - t > 2g = {\color{black}u^2 - u} $, {\color{black}and let $ \Delta \subseteq [r,n] $ with $ r \in \Delta $}. The corresponding secret sharing scheme in Corollary \ref{corollary SSS one-point construction2} constructed using $ \mathcal{X} $ has information rate $ \ell / n = (r - t - q + \sqrt{q})/ (\sqrt{q}^3) $, reconstruction $ r $, privacy $ t $ and, for any {\color{black}$ \delta \in \Delta $} and any $ I \subseteq [n] $ of size $ \delta $, it holds that
$$ {\rm CO}(I) = \frac{\ell (t + q - \sqrt{q})}{\delta - t - q + \sqrt{q}}, \quad \textrm{or} \quad {\rm DB}(I) = \frac{\ell \delta}{\delta - t - q + \sqrt{q}}. $$
\end{proposition}

These schemes can be explicitly constructed using explicit descriptions of the generator matrices of the involved Hermitian codes, which can be found in \cite{notehermitian, tiersma}.

We now give {\color{black}an example of attainable parameters and an explicit toy example of Construction 1 using Hermitian codes: }

\begin{example} \label{example hermitian parameters}
Take the field $ \mathbb{F}_q $ with $ q = 16 $, and let the information rate be $ \frac{1}{2} $.

The maximum length obtained by a Reed-Solomon code is $ n=16 $. Then $ \ell = 8 $ and $ r-t = 8 $. In this case, the decoding bandwidth for sets of size $ \delta $ is
$$ {\rm DB} = \frac{8d}{8 + ( \delta -r)} = \frac{ \delta }{1 + ( \delta -r)/8}. $$

{\color{black}
\begin{figure}[t]
\centering
\footnotesize
\begin{tabular}{cccccccccc}
\hline
t & 0 & 1 & 2 & 3 & 4 & 5 & 6 & 7 & 8 \\
\hline 
r/n & 0.50 & 0.56 & 0.63 & 0.69 & 0.75 & 0.81 & 0.88 & 0.94 & 1 \\
\hline 
{\rm DB}/n & 0.50 & 0.53 & 0.57 & 0.62 & 0.67 & 0.73 & 0.80 & 0.89 & 1 \\
\hline \hline
t & -- & -- & -- & 0 & 4 & 8 & 12 & 16 & 20 \\
\hline 
r/n & -- & -- & -- & 0.69 & 0.75 & 0.81 & 0.88 & 0.94 & 1 \\
\hline 
{\rm DB}/n & -- & -- & -- & 0.62 & 0.67 & 0.73 & 0.80 & 0.89 & 1 \\
\hline
\end{tabular} \\
\caption{Values of $ t $, $ r/n $ and $ {\rm DB}/n $, for $ \delta =n $, using Reed-Solomon codes of length $ n=16 $ and Hermitian codes of length $ n=64 $, respectively, both with $ q=16 $.}
\label{fig hermitian}
\end{figure}
}

The maximum length obtained by a Hermitian code is $ n=64 $, with genus $ g=6 $. Then $ \ell = 32 $ and $ r-t = \ell+2g = 44 $. In this case, the decoding bandwidth for sets of size $ \delta $ is
$$ {\rm DB} = \frac{32d}{32 + ( \delta -r)} = \frac{ \delta }{1 + ( \delta -r)/32}. $$

Fig. \ref{fig hermitian} shows several values of $ t $, {\color{black}$ r/n $} and $ {\rm DB}/n $, for $ \delta =n $, when using Reed-Solomon codes and Hermitian codes of lengths $ n=16 $ and $ n=64 $, respectively. {\color{black}As observed in Subsection \ref{subsec Cons 1 and 2 for AG codes}, the defect with the optimal case is only on $ t $, while the decoding bandwidths present the same improvements with respect to $ r $. }
\end{example}

{\color{black}
\begin{example} \label{example generator matrix Hermitian}
We will use here the worked out example of Hermitian codes from \cite{tiersma}. Let $ q = 4 $ and $ \mathbb{F}_4 = \{ 0, 1, \omega, \overline{\omega} \} $ with $ \omega \overline{\omega} = \omega + \overline{\omega} = 1 $. We may consider that the Hermitian curve has equation $ x^3 + y^3 + z^3 = 0 $ after a transformation, and it has the $ 9 $ rational points $ Q = (0,1,1) $, $ P_1 = (1,0,\overline{\omega}) $, $ P_2 = (1,0,\omega) $, $ P_3 = (1,0,1) $, $ P_4 = (1, \overline{\omega}, 0) $, $ P_5 = (1, \omega, 0) $, $ P_6 = (1,1,0) $, $ P_7 = (0, \overline{\omega}, 1) $, $ P_8 = (0, \omega, 1) $. In such case, the genus is $ g = 1 $ and the codes $ \mathcal{C}(D, \mu Q) $, for $ \mu = 1,2,3,4 $, where $ D = \sum_{i=1}^8 P_i $, are $ \mu $-dimensional and generated by the first $ \mu $ rows of the matrix
$$ \left(
\begin{array}{cccccccc}
1 & 1 & 1 & 1 & 1 & 1 & 1 & 1 \\
\omega & \overline{\omega} & 1 & \omega & \overline{\omega} & 1 & 0 & 0 \\
0 & 0 & 0 & 1 & 1 & 1 & \omega & \overline{\omega} \\
\overline{\omega} & \omega & 1 & \overline{\omega} & \omega & 1 & 0 & 0 
\end{array}
\right). $$
Choose now $ \mu_2 = 2 $, $ \mu_1 = 3 $ and $ \mu = 4 $. Then by Proposition \ref{proposition using AG for construction 1}, $ n = 8 $, $ t = 1 $, $ r = 4 $, $ \ell = 1 $, $ \delta = 5 $, and $ \alpha = \mu - \mu_2 = 2 $. To obtain the shares of the corresponding scheme, we need to multiply the previous matrix on the left by
$$ \left( \begin{array}{cc|c|c}
r_{11} & r_{12} & s_{11} & s_{21} \\
\hline
r_{21} & r_{22} & s_{21} & 0
\end{array}
\right) , $$
where the secret is $ (s_{11}, s_{21})^T \in \mathbb{F}_4^{2 \times 1} $. For instance, the first share is
$$ (r_{11} + r_{12} \omega + s_{21} \overline{\omega} , r_{21} + r_{22} \overline{\omega} )^T \in \mathbb{F}_4^{2 \times 1}. $$
When downloading $ \delta=4 $ shares, we only need to download the first of the two components of each share. The communication overheads for sets of sizes $ r $ and $ \delta $ are then
$$ {\rm CO}(r) = 3, \quad \textrm{and} \quad {\rm CO}(\delta) = 3/2. $$
Hence communication overheads are reduced by half when contacting $ \delta=4 $ shares instead of $ r=3 $.
\end{example}
}

\subsection{Asymptotic behaviour for a fixed finite field} \label{subsec asymptotic}

In this subsection, we use the secret sharing schemes based on one-point AG codes in Corollary \ref{corollary SSS one-point construction2} to obtain sequences of schemes that are communication efficient and whose lengths go to infinity while being defined over a fixed finite field $ \mathbb{F}_q $. {\color{black}We treat only Construction 1 for simplicity on the asymptotic parameters, being Construction 2 analogous}. In this way, we show that the defect in the decoding bandwidth in the schemes in Corollary \ref{corollary SSS one-point construction2} with respect to the optimal values given by (\ref{eq lower bound DB}) depends only on the field size $ q $ and not on the length $ n $, meaning that, for a given defect, we may fix a suitable field size, and let the lengths be arbitrarily large.

As observed in Remark \ref{remark near optimal low g}, Corollary \ref{corollary SSS one-point construction2} states the existence of communication efficient schemes depending on the existence of irreducible projective curves $ \mathcal{X} $ over $ \mathbb{F}_q $ with as many rational points and small genus as possible {\color{black}(and explicit schemes can be constructed if data about $ \mathcal{X} $ is known)}. Therefore, it will be essential to make use of Ihara's constant \cite{ihara}
\begin{equation}
A(q) = \limsup_{g(\mathcal{X}) \rightarrow \infty} \frac{N(\mathcal{X})}{g(\mathcal{X})},
\label{definition iharas constant}
\end{equation}
where the limit is taken over all irreducible projective curves $ \mathcal{X} $ over $ \mathbb{F}_q $ of genus $ g(\mathcal{X}) > 0 $, and where $ N(\mathcal{X}) $ denotes the number of rational points in $ \mathcal{X} $. Serre's lower bound \cite{serre} and the Drinfeld-Vl\u{a}du\c{t} upper bound \cite{vladut} state that 
\begin{equation}
c \log(q) \leq A(q) \leq \sqrt{q}-1, \label{eq vladut bound ihara}
\end{equation}
for a constant $ c > 0 $ that does not depend on $ q $, and where the equality $ A(q) = \sqrt{q}-1 $ holds if $ q $ is a perfect square \cite{ihara}. See also \cite[Section 2.9]{handbook}, \cite[Section 12.7.7]{cramerbook} and \cite[Chapter 7]{stichtenothbook} for more details on Ihara's constant and asymptotic behaviour of AG codes.

We will consider sequences of irreducible projective curves $ (\mathcal{X}_i)_{i=1}^\infty $ such that $ N(\mathcal{X}_i) \longrightarrow \infty $ and
\begin{equation}
\lim_{i \rightarrow \infty} \frac{N(\mathcal{X}_i)}{g(\mathcal{X}_i)} = A(q). \label{eq good sequences}
\end{equation} 

We may now state the following asymptotic consequence of Corollary \ref{corollary SSS one-point construction2}:

\begin{proposition} \label{proposition asymptotics}
For any $ 0 \leq T < R {\color{black}< D} \leq 1 $ with $ R - T > 2/A(q) $, there exists a strictly increasing sequence of positive integers $ (n_i)_{i=1}^\infty $ and a sequence of secret sharing schemes defined over $ \mathbb{F}_q $ such that, for large enough $ i $, the $ i $-th scheme has length $ n_i $ {\color{black}and the following parameters:
\begin{enumerate}
\item
Information rate $ \ell_i / n_i \geq L = R - T - 2/A(q) $,
\item
Reconstruction $ r_i = \lceil R n_i \rceil $,
\item
Privacy $ t_i = \lfloor T n_i \rfloor $,
\end{enumerate}
and for any set $ I \subseteq [n_i] $ of size $ \delta_i = \lceil D n_i \rceil $, it holds that
$$ \frac{L D}{D - T} \leq DB(I) \leq \frac{L D}{D - T - 2/A(q)}. $$
As in Corollary \ref{corollary SSS one-point construction2}, each scheme in the sequence can be explicitly constructed if expressions of the rational places and Riemann-Roch spaces are known. }
\end{proposition}
\begin{proof}
Take a sequence of irreducible projective curves $ (\mathcal{X}_i)_{i=1}^\infty $ satisfying {\color{black}$ N(\mathcal{X}_i) \longrightarrow \infty $ and} (\ref{eq good sequences}). {\color{black} Let $ 2 g(\mathcal{X}_i) + 2 \leq n_i \leq N(\mathcal{X}_i) - 1 $ be such that $ \lim_{i \rightarrow \infty} n_i / g(\mathcal{X}_i) = A(q) $. The parameters} $ r_i = \lceil R n_i \rceil $ and $ t_i = \lfloor T n_i \rfloor $ satisfy $ r_i - t_i > 2 g(\mathcal{X}_i) $ for $ i $ large enough. The result then follows by defining divisors and one-point AG codes as in the proof of Corollary \ref{corollary SSS one-point construction2}.
\end{proof}

\begin{remark}
Observe that (\ref{eq vladut bound ihara}) implies that, as $ q \longrightarrow \infty $, it holds that $ 2/A(q) \longrightarrow 0 $, and thus by Proposition \ref{proposition asymptotics}, it follows that the asymptotic decoding bandwidth can be as close to the optimal value as wanted for large enough finite fields. Formally, for a fixed $ \varepsilon >0 $, we may fix $ q $ such that the difference with the optimal value is $ \leq \varepsilon $, and then let the lengths be arbitrary large.

Observe also that the lengths of the schemes in the previous proposition are lower bounded by (twice) the genus of the corresponding curves. We leave open how to find more flexible (shorter) schemes.
\end{remark}

{\color{black}
We conclude this subsection by observing that the exponentiating parameter $ \alpha $ grows linearly with $ n $ in Construction 1 when using algebraic geometry codes. This leads to the discussion in item 2 in Subsection \ref{subsec our motivations} on the share sizes of the schemes in \cite{rawad, efficient} and those based on algebraic geometry codes.

Fix an information rate $ 0 < L < 1 $. Construction 1 in \cite{rawad} requires $ \alpha = \delta - t = \lambda n $, where $ L < \lambda < 1 $ since $ \delta-t > r-t = \ell $. On the other hand, Construction 1 based on a sequence of irreducible projective curves attaining equality in (\ref{definition iharas constant}) satisfies $ \alpha = \lambda n $ for large lengths, where $ L + 2 / A(q) < \lambda < 1 $. Hence the shares are of size at least $ n^{\lambda n} $, that is $ \Omega(n\log(n)) $ bits, in the first case, whereas they are of size $ O(q^{\lambda n}) $, that is $ O(n) $ bits, in the second case.

}

\subsection{Using Garc{\'i}a-Stichtenoth's second tower}

In this subsection, we apply Proposition \ref{proposition asymptotics} to the so-called second tower of projective curves $ (\mathcal{X}_i)_{i=1}^\infty $ introduced by Garc{\'i}a and Stictenoth in \cite{garciatower}. This sequence of projective curves is defined over a field whose size $ q $ is a perfect square and satisfies (\ref{eq good sequences}), where $ A(q) = \sqrt{q}-1 $, hence being asymptotically optimal. 

Moreover, they are explicitly defined in \cite{garciatower}, and there are efficient algorithms to construct the {\color{black}generator matrices of the} corresponding algebraic geometry codes, as shown in \cite{shum} (the $ i $-th algorithm has complexity $ O(n_i^3 \log(n_i)^3) $ or $ O(n_i^3) $ in some cases, being $ n_i $ the length of the $ i $-th code {\color{black}\cite[Theorem 7]{shum}}). We recall here that this sequence of codes is one of the few {\color{black}explicitly} known sequences that has better asymptotic parameters than the existencial Gilbert-Varshamov bound.

Assume again that the field size is a perfect square $ q = u^2 $, and consider Garc{\'i}a and Stictenoth's second tower of projective curves $ (\mathcal{X}_i)_{i=1}^\infty $ from \cite{garciatower}. The $ i $-th curve $ \mathcal{X}_i $ has $ N(\mathcal{X}_i) > u^i(u-1) $ rational points, and its genus is given by 
\begin{displaymath}
g(\mathcal{X}_i) = \left\lbrace 
\begin{array}{ll}
(u^{\frac{i}{2}} - 1)^2 & \textrm{if } i \textrm{ is even}, \\
(u^{\frac{i+1}{2}} - 1)(u^{\frac{i-1}{2}} - 1) & \textrm{if } i \textrm{ is odd}. \\
\end{array} \right.
\end{displaymath}

Specializing Proposition \ref{proposition asymptotics} to this tower of curves, we obtain the following result:

\begin{proposition} 
For any finite field $ \mathbb{F}_q $ whose size $ q $ is a perfect square, and for any $ 0 \leq T < R {\color{black}< D} \leq 1 $ with $ R - T > 2/(\sqrt{q}-1) $, there exists a strictly increasing sequence of positive integers $ (n_i)_{i=1}^\infty $ and a sequence of secret sharing schemes defined over $ \mathbb{F}_q $ such that, for large enough $ i $, the $ i $-th scheme has length $ n_i $ {\color{black}and the following parameters:
\begin{enumerate}
\item
Information rate $ \ell_i / n_i \geq L = R - T - 2/(\sqrt{q}-1) $,
\item
Reconstruction $ r_i = \lceil R n_i \rceil $,
\item
Privacy $ t_i = \lfloor T n_i \rfloor $,
\end{enumerate}
and for any set $ I \subseteq [n_i] $ of size $ \delta_i = \lceil D n_i \rceil $, it holds that
$$ \frac{L D}{D - T} \leq DB(I) \leq \frac{L D}{D - T - 2/(\sqrt{q}-1)}. $$
Moreover, generator matrices for the $ i $-th scheme can be constructed explicitly with complexity $ O(n_i^3 \log(n_i)^3) $ following the algorithm in \cite[Subsection IV-A]{shum}.}
\end{proposition}
\begin{proof}
{\color{black}The proof is exactly the same as that of Proposition \ref{proposition asymptotics}, although using Garc{\'i}a and Stictenoth's second tower of projective curves, as described in this subsection. We recall that, according to \cite[Theorem 7]{shum}, to construct generator matrices for the $ i $-th scheme with the mentioned complexity, we need to choose the length of the $ i $-th scheme as $ n_i = u^i(u-1) $, which satisfies the constraints in the proof of Proposition \ref{proposition asymptotics}. Recall also the definitions of the parameters for the codes corresponding to the $ i $-th scheme from the proof of Corollary \ref{corollary SSS one-point construction2}. }
\end{proof}

\section{Strongly secure and communication efficient schemes}  \label{sec strongly secure}

In this section, we see how to obtain strongly secure secret sharing schemes as in \cite[Section 4]{kurihara-secret} that are also communication efficient at the same time. As explained in Section \ref{sec intro}, strongly secure schemes were introduced in \cite{yamamoto-threshold} and allow to keep {\color{black}components} of the secret safe even when more than $ t $ shares are eavesdropped. This feature makes them behave like perfect schemes in terms of {\color{black}component-wise} security, while having the {\color{black}security and} efficiency in storage of ramp schemes. {\color{black}In this sense, strongly secure schemes provide a strict improvement on ramp schemes and an interesting alternative to perfect schemes.} 

{\color{black}
\subsection{Strong security}  \label{subsec strong security}

}

Although strongly secure secret sharing schemes were introduced in \cite{yamamoto-threshold}, we will use the extended definition from \cite[Definition 18]{kurihara-secret}:

\begin{definition} [\textbf{\cite{kurihara-secret, yamamoto-threshold}}]
{\color{black}Let $ \sigma \in \mathbb{N} $.} We say that a secret sharing scheme $ F: \mathcal{A}^\ell \longrightarrow \mathcal{A}^n $ is $ \sigma $-strongly secure if, for all $ I \subseteq [\ell] $ and all $ J \subseteq [n] $ with $ \# I + \# J \leq \sigma + 1 $, it holds that
$$ I(\mathbf{s}_I; \mathbf{x}_J) = 0, $$
where $ \mathbf{s} $ denotes the random variable corresponding to the secret in $ \mathcal{A}^\ell $, and $ \mathbf{x} = F(\mathbf{s}) $ denotes the random variable corresponding to the shares in $ \mathcal{A}^n $.

Observe that, if a scheme is $ \sigma $-strongly secure and $ \sigma^\prime \leq \sigma $, then it is also $ \sigma^\prime $-strongly secure. Hence {\color{black}it is convenient to} define the maximum strength of a scheme $ F: \mathcal{A}^\ell \longrightarrow \mathcal{A}^n $ as
$$ \sigma_{\rm max}(F) = \max \{ \sigma \in \mathbb{N} \mid F \textrm{ is } \sigma \textrm{-strongly secure} \}. $$
\end{definition}

Observe that the reconstruction and privacy threshold values are also monotonous. Hence, we define for convenience the minimum reconstruction and maximum privacy thresholds, respectively, of a scheme $ F: \mathcal{A}^\ell \longrightarrow \mathcal{A}^n $ as
$$ r_{\rm min}(F) = \min \{ r \in \mathbb{N} \mid F \textrm{ has } r \textrm{-reconstruction} \}, \textrm{ and} $$
$$ t_{\rm max}(F) = \max \{ t \in \mathbb{N} \mid F \textrm{ has } t \textrm{-privacy} \}. $$

We next give upper bounds on $ \sigma_{\rm max}(F) $ in terms of $ r_{\rm min}(F) $ and $ t_{\rm max}(F) $ that follow easily from the definitions. They will allow us to claim the optimality of the construction in Corollary \ref{corollary strongly secure MDS}, which we will present in Subsection \ref{subsec optimally strongly secure MDS}.

\begin{lemma}
For a secret sharing scheme $ F: \mathcal{A}^\ell \longrightarrow \mathcal{A}^n $ , it holds that
\begin{equation}
\sigma_{\rm max}(F) \leq t_{\rm max}(F) + \ell - 1 \leq r_{\rm min}(F) - 1. \label{eq upper bound strength}
\end{equation}
\end{lemma}
\begin{proof}
Take $ I = [\ell] $. By definition, if the scheme has $ \sigma $-strong securiry, then it holds that $ I(\mathbf{s}_I; \mathbf{x}_J) = 0 $ for any $ J \subseteq [n] $ of size $ \sigma + 1 - \ell $. Hence the scheme $ F $ has $ (\sigma + 1 - \ell) $-privacy. Since $ t_{\rm max}(F) $ is the maximum privacy threshold, it holds by definition that $ \sigma_{\rm max}(F) + 1 - \ell \leq t_{\rm max}(F) $, hence the first bound follows. The second bound follows directly from (\ref{eq lower bound info rate}).
\end{proof}

Observe that a strongly secure secret sharing scheme not only satisfies that $ \sigma_{\rm max}(F)- \ell + 1 $ is a privacy threshold but, in addition, if more than $ \sigma_{\rm max}(F) - \ell + 1 $ shares are leaked and some information about the secret is obtained by an eavesdropper, we are guaranteed that no information about any collection of $ \mu $ components of the secret $ \mathbf{s} $ are leaked if the number of leaked shares is at most $ \sigma_{\rm max}(F) - \mu + 1 $, for $ \mu \leq \ell $. In particular, no information about any single component $ s_i \in \mathcal{A} $ of the secret is leaked if at most $ \sigma_{\rm max}(F) $ shares are obtained by the eavesdropper.

In the optimal case $ \sigma_{\rm max}(F) = t_{\rm max}(F) + \ell - 1 = r_{\rm min}(F) - 1 $, if $ t_{\rm max}(F) + \mu $ shares are eavesdropped, for some $ \mu > 0 $, then no information is leaked about any collection of $ \ell - \mu $ components of $ \mathbf{s} $, $ \mu \leq \ell $. Observe that in particular, no information is leaked about any component $ s_i \in \mathcal{A} $ if $ r_{\rm min}(F) - 1 $ shares are eavesdropped, which is the same amount as in an optimal perfect scheme. However, the scheme $ F : \mathcal{A}^\ell \longrightarrow \mathcal{A}^n $ is also an optimal ramp scheme, hence having the {\color{black}component-wise} security advantages of perfect schemes and the {\color{black}security and} efficiency in storage of ramp schemes{\color{black}, as explained at the beginning of this section}.

\subsection{Massey-type secret sharing schemes} \label{subsec massey-type schemes}

We will consider Construction 2 combined with Massey-type secret sharing schemes, which are a modification of the schemes in \cite{massey}. We omit the details regarding Construction 1 for brevity, since they are analogous.

{\color{black}To consider strong security, we will see the $ i $-th component of the secret $ S \in \mathbb{F}_q^{\alpha \times \ell} $ as its $ i $-th column $ S_i \in \mathbb{F}_q^{\alpha \times 1} $, for $ i = 1,2, \ldots, \ell $, which is a symbol in the alphabet $ \mathcal{A} = \mathbb{F}_q^\alpha $. 

The idea behing the Massey construction is to use codes obtained by restricting longer ones to certain coordinates. In our case, we will fix throughout the section positive integers $ h, \ell \leq n $, and linear codes
\begin{equation}
\mathcal{D}_1 = \mathcal{D}^{(h)} \subsetneqq \mathcal{D}^{(h-1)} \subsetneqq \ldots \subsetneqq \mathcal{D}^{(2)} \subsetneqq \mathcal{D}^{(1)} \subseteq \mathbb{F}_q^{\ell + n}.
\label{eq longer codes massey}
\end{equation}
We may now define our modified Massey-type construction of secret sharing schemes as follows (recall the definition of restricted and shortened codes from Subsection \ref{subsec notation}):
}

\begin{definition} \label{definition massey schemes}
{\color{black}Assume that the codes in (\ref{eq longer codes massey}) satisfy the following:
\begin{enumerate}
\item
$ \mathcal{D}_{1 [\ell]} = \mathbb{F}_q^\ell $, and
\item
$ \dim(\mathcal{D}_{[\ell + 1, \ell + n]}^{(1)}) = \dim(\mathcal{D}^{(1)}) $.
\end{enumerate}
Define now the codes}
\begin{equation}
\begin{split}
\mathcal{C}_2 = \mathcal{D}_1^{[\ell+1, \ell+n]} & \subsetneqq \mathcal{C}_1 = \mathcal{D}_{1 [\ell+1, \ell+n]}, \quad \textrm{and} \\
 \mathcal{C}^{(j)} & = \mathcal{D}^{(j)}_{[\ell+1, \ell+n]}, 
\end{split}
\label{eq massey type codes}
\end{equation} 
for $ j = 1,2, \ldots, h $. We say that the secret sharing scheme in Construction 2 is of Massey-type if it is constructed using such linear codes {\color{black}
$$ \mathcal{C}_2 \subsetneqq \mathcal{C}_1 = \mathcal{C}^{(h)} \subsetneqq \mathcal{C}^{(h-1)} \subsetneqq \ldots \subsetneqq \mathcal{C}^{(2)} \subsetneqq \mathcal{C}^{(1)} \subseteq \mathbb{F}_q^n. $$ }
\end{definition}

{\color{black}The following is the main result of this section. It states that the secret sharing scheme in Construction 2 based on the codes in the previous definition not only has the properties of the scheme in Theorem \ref{theorem general univeral staircase properties}, but in addition we may give a lower bound on its strong security that will be useful in the next subsections. The proof is given in Appendix \ref{app proof strongly secure}. }

\begin{theorem} \label{theorem strongly secure construction}
{\color{black}Let $ \Delta \subseteq [n] $ with elements $ 0 < \delta_h < \delta_{h-1} < \ldots < \delta_2 < \delta_1 \leq n $, and assume that the codes in (\ref{eq massey type codes}) satisfy that 
$$ \delta_j \geq n - d(\mathcal{C}^{(j)}, \mathcal{C}_2) + 1, $$
for $ j = 1,2, \ldots, h $. Then the secret sharing scheme $ F : \mathcal{A}^\ell \longrightarrow \mathcal{A}^n $ from Construction 2 using the codes in (\ref{eq massey type codes}) has information rate $ \ell / n $, reconstruction $ r = \delta_h $, privacy $ t = d(\mathcal{C}_2^\perp, \mathcal{C}^{(1) \perp}) - 1 $ and communication overheads as in Theorem \ref{theorem general univeral staircase properties}, where it holds that $ k_2 = \dim(\mathcal{C}_2) = \dim(\mathcal{C}_1) - \ell $ and $ k^{(j)} = \dim(\mathcal{D}^{(j)}) = \dim(\mathcal{C}^{(j)}) $, for $ j = 1,2, \ldots, h $. In addition, it holds that
$$ \sigma_{\rm max}(F) \geq \min \{ d(\mathcal{G}_i^\perp) \mid 1 \leq i \leq \ell \} - 1, $$
where we define the linear codes
$$ \mathcal{G}_i = \mathcal{D}_1^{[\ell + n] \setminus \{ i \}}, $$
for $ i = 1,2, \ldots , \ell $.}
\end{theorem}
\begin{proof}
{\color{black}See Appendix \ref{app proof strongly secure}. }
\end{proof}

\subsection{Optimal strongly secure and communication efficient schemes based on MDS codes} \label{subsec optimally strongly secure MDS}

In this subsection, we give the first construction of a ramp secret sharing scheme that has optimal threshold parameters with respect to its information rate (\ref{eq lower bound info rate}), optimal strong security in the sense of (\ref{eq upper bound strength}) and optimal communication efficiency in the sense of (\ref{eq lower bound CO}) and (\ref{eq lower bound DB}). The result is a consequence of Theorem \ref{theorem strongly secure construction}:

\begin{corollary} \label{corollary strongly secure MDS}
{\color{black}Fix positive integers $ n $, $ r $ and $ t $ such that $ 0 \leq t < r \leq n $, define $ \ell = r-t $, {\color{black}and let $ \Delta \subseteq [r,n] $ with $ r \in \Delta $}. If $ \ell + n \leq q $, then we may choose the codes in (\ref{eq longer codes massey}) as MDS codes with $ \dim(\mathcal{D}_1) = r $, and satisfying the conditions in Definition \ref{definition massey schemes} and Theorem \ref{theorem strongly secure construction}.

In such case, the secret sharing scheme $ F : \mathcal{A}^\ell \longrightarrow \mathcal{A}^n $ from Construction 2 using the codes in (\ref{eq massey type codes}) has information rate $ \ell / n $, reconstruction $ r  $, privacy $ t $ (hence $ \ell = r-t $ is optimal by (\ref{eq lower bound info rate})), optimal decoding bandwidth (meaning equality in (\ref{eq lower bound CO}) and (\ref{eq lower bound DB})) for any {\color{black}$ \delta \in \Delta $} and any set $ I \subseteq [n] $ of size $ \delta $, and in addition it has maximum strength
$$ \sigma_{\rm max}(F) = t + \ell - 1 = r - 1, $$
which is optimal by (\ref{eq upper bound strength}). In particular, it holds that
$$ I(S_I; X_J) = 0, $$
for all $ I \subseteq [\ell] $ and all $ J \subseteq [n] $ with $ \# I + \# J \leq r $. }
\end{corollary}
\begin{proof}
{\color{black}First if $ \ell + n \leq q $, then we may take the codes in (\ref{eq longer codes massey}) as nested Reed-Solomon codes, which are MDS, with $ \dim(\mathcal{D}_1) = r $ and satisfying the conditions in Definition \ref{definition massey schemes} and Theorem \ref{theorem strongly secure construction}.

Now by Theorem \ref{theorem strongly secure construction}, we only} need to show that $ d(\mathcal{G}_i^\perp) \geq {\color{black}r} $, for  $ i = 1, 2, \ldots, \ell $. Since $ \mathcal{D}_1 $ is MDS, the linear code $ \mathcal{G}_i $ is again MDS, which implies in turn that $ \mathcal{G}_i^\perp $ is MDS. The length of this latter code is $ \ell + n -1 $ and its dimension is {\color{black}$ \dim(\mathcal{G}_i^\perp) = \dim(\mathcal{D}_1^\perp) = \ell + n - r $}, therefore $ d(\mathcal{G}_i^\perp) = {\color{black}r} $, and we are done.
\end{proof}

Observe that the requirement on the field size ($ q \geq \ell + n $) is the same as in previous optimal strongly secure schemes \cite{nishiara}, and the expansion of the alphabet size is the same as in previous optimal communication efficient schemes \cite{rawad, efficient}. 

For each component of the secret $ S_i \in \mathcal{A} $, the scheme behaves as an optimal perfect scheme (as Shamir's scheme), and for the entire secret $ S \in \mathcal{A}^\ell $, it behaves as an optimal ramp scheme. Therefore, it has the {\color{black}component-wise} security advantages of the first and the {\color{black}security and} storage efficiency of the second, while being communication efficient at the same time.

\subsection{Strongly secure and communication efficient schemes based on algebraic geometry codes} \label{subsec massey AG codes}

In this subsection, we extend Corollary \ref{corollary strongly secure MDS} by using general AG codes. In this way, we overcome the limitation on the field size in the previous construction, at the cost of near optimality. {\color{black}Analogous strongly secure schemes based on AG codes, but without communication efficiency, have been given in \cite{ryutaroh-multi}.}

We give a construction analogous to that in Corollary \ref{corollary SSS one-point construction2}:

\begin{corollary} \label{corollary strongly secure AG}
If there exists an irreducible projective curve $ \mathcal{X} $ over $ \mathbb{F}_q $ with $ N \geq \ell + n + 1 $ rational points and genus $ g $, then for any $ 0 \leq t < r \leq n $ with {\color{black}$ \ell = r - t - 2g $, and any $ \Delta \subseteq [r,n] $ with $ r \in \Delta $}, there exists a secret sharing scheme $ F : \mathcal{A}^\ell \longrightarrow \mathcal{A}^n $ with information rate $ \ell / n $, reconstruction $ r $, privacy $ t $ and, for any {\color{black}$ \delta \in \Delta $} and any $ I \subseteq [n] $ of size $ \delta $, it holds that
$$ {\rm CO}(I) = \frac{\ell (t + 2g)}{\delta - t - 2g}, \quad \textrm{or} \quad {\rm DB}(I) = \frac{\ell \delta}{\delta - t - 2g}. $$
In addition, the scheme has maximum strength
$$ \sigma_{\rm max}(F) \geq t + \ell - 1 = r - 2g - 1. $$
In particular, it holds that
$$ I(S_I; X_J) = 0, $$
for all $ I \subseteq [\ell] $ and all $ J \subseteq [n] $ with $ \# I + \# J \leq t + \ell = r - 2g $. 

{\color{black}Moreover, the scheme can be explicitly constructed if expressions of the rational places and Riemann-Roch spaces corresponding to $ \mathcal{X} $ are known. }
\end{corollary}
\begin{proof}
As in the proof of Corollary \ref{corollary SSS one-point construction2}, take pairwise distinct rational points $ Q, P_1, P_2, \ldots, P_{\ell + n} \in \mathcal{X} $, and define $ \mu_1 = r - 1 $ and $ \mu^{(j)} = \delta_j - 1 $, for $ j = 1,2, \ldots, h $, {\color{black}where the elements in $ \Delta $ are $ r = \delta_h < \delta_{h-1} < \ldots < \delta_2 < \delta_1 \leq n $}. Define $ E = P_1 + P_2 + \cdots + P_{\ell + n} $, and the AG codes
\begin{equation*}
\mathcal{D}_1 = \mathcal{C}(E, \mu_1 Q), \quad \textrm{and} \quad \mathcal{D}^{(j)} = \mathcal{C}(E, \mu^{(j)} Q),
\end{equation*}
for $ j = 1,2, \ldots, h $. Since $ 2g - 2 < \mu_1 < n $, it holds that $ \dim(\mathcal{D}_1) = \mu_1 - g + 1 \geq \ell $ by Lemma \ref{lemma dimensions AG}. Hence by rearranging the points $ P_i $ if necessary, we may assume that $ \mathcal{D}_{1 [\ell]} = \mathbb{F}_q^\ell $. Moreover, we have that {\color{black}$ \mathcal{D}^{(1)}_{[\ell + 1, \ell + n]} = \mathcal{C}(D, \mu^{(1)} Q) $, where $ D = P_{\ell + 1} + P_{\ell + 2} + \cdots + P_{\ell + n} $. Thus
$$ \dim(\mathcal{D}^{(1)}_{[\ell + 1, \ell + n]}) = \dim(\mathcal{D}^{(1)}) $$
by} Lemma \ref{lemma dimensions AG} since $ 2g - 2 < \mu^{(1)} < n < \ell + n $. Therefore the assumptions in Definition \ref{definition massey schemes} are satisfied. 

On the other hand, {\color{black}the codes in (\ref{eq massey type codes}) are given by }
$$ \mathcal{C}_2 = \mathcal{C}(D, \mu_1 Q - P_1 - P_2 - \cdots - P_\ell) \subsetneqq \mathcal{C}_1 = \mathcal{C}(D, \mu_1 Q), $$
$$ \textrm{and} \quad \mathcal{C}^{(j)} = \mathcal{C}(D, \mu^{(j)} Q), $$
for $ j = 1,2, \ldots, h $, where $ \deg(\mu_1 Q - P_1 - P_2 - \cdots - P_\ell) = \mu_1 - \ell = t + 2g - 1 $. Therefore, the assumptions in Proposition \ref{proposition using AG for construction 2} are satisfied and hence the result follows, except for the claim on the maximum strength, which we now prove. For a fixed $ i \in [\ell] $, observe that
$$ \mathcal{G}_i = \mathcal{C}(P_1 + P_2 + \cdots + \widehat{P}_i + \cdots + P_{\ell + n}, \mu_1 Q - P_i). $$
Since $ 2g - 2 < \mu_1 - 1 < \ell + n - 1 $, it follows that $ \dim(\mathcal{G}_i) = \mu_1 - g $ by Lemma \ref{lemma dimensions AG}. Therefore, it holds that
$$ d(\mathcal{G}_i^\perp) \geq \mu_1 - 2g + 1 $$
by Lemma \ref{lemma goppa}. Hence the scheme has maximum strength
$$ \sigma_{\rm max}(F) \geq \mu_1 - 2g = r - 2g - 1 = t + \ell - 1 $$
by Theorem \ref{theorem strongly secure construction}, and we are done.
\end{proof}

In the same way as in Subsection \ref{subsec asymptotic}, we may give the following asymptotic consequence, which shows again that the defect in the decoding bandwidth with respect to the optimal case depends only on the field size $ q $ and not on the length $ n $. Again, this means that we may fix a suitable field size for a given defect, and then let the lengths be arbitrary large.

\begin{corollary} \label{corollary asymptotics}
For any $ 0 \leq T < R {\color{black}< D} \leq 1 $ with $ R - T > 4/A(q) $, there exists a strictly increasing sequence of positive integers $ (n_i)_{i=1}^\infty $ and a sequence of secret sharing schemes defined over $ \mathbb{F}_q $ such that, for large enough $ i $, the $ i $-th scheme has length $ n_i $ {\color{black}and the following parameters:
\begin{enumerate}
\item
Information rate $ \ell_i / n_i \geq L = R - T - 4/A(q) $,
\item
Reconstruction $ r_i = \lceil R n_i \rceil $,
\item
Privacy $ t_i = \lfloor T n_i \rfloor $,
\end{enumerate}
and for any set $ I \subseteq [n_i] $ of size $ \delta_i = \lceil D n_i \rceil $, it holds that
$$ \frac{L D}{D - T} \leq DB(I) \leq \frac{L D}{D - T - 4/A(q)}. $$
In addition, the $ i $-th scheme has maximum strength 
$$ \sigma_i = t_i + \ell_i - 1. $$
As in Corollary \ref{corollary SSS one-point construction2}, each scheme in the sequence can be explicitly constructed if expressions of the rational places and Riemann-Roch spaces are known. }
\end{corollary}
\begin{proof}
The proof is the same as that of Proposition \ref{proposition asymptotics}. However, we now need that $ \ell_i + n_i + 1 \leq N_i $, where $ N_i $ is the number of rational points in the $ i $-th projective curve $ \mathcal{X}_i $. Since $ \ell_i \leq n_i $, we may take $ N_i \geq 2n_i + 1 $ and hence $ g_i/n_i \longrightarrow 2/A(q) $.
\end{proof}

We may also give analogous examples as in Section \ref{sec AG codes} using Hermitian codes and Garc{\'i}a and Stichtenoth's second tower of projective curves. The only adjustment is that we have to substitute the number of rational points $ N $ by $ n+\ell + 1 $ instead of $ n+1 $. We leave the details to the reader.

\section{Conclusion}

In this paper, we have given a new framework to construct communication efficient secret sharing schemes based on sequences of nested linear codes, extending the previous Shamir-type communication efficient schemes from the literature. We have given two general constructions, one with small alphabet but non-universal low decoding bandwidths, and one with large alphabet but universal low decoding bandwidths. 

By specializing the codes to algebraic geometry codes, we have obtained communication efficient secret sharing schemes with low decoding bandwidths and large lengths for a fixed finite field, in contrast with previous works. The obtained near optimality implies that, for a given deffect with the optimal case, we may fix a suitable field size and let the lengths be arbitrary large. Moreover, we have seen that the loss is on the privacy thresholds, as in previous schemes based on algebraic geometry codes, whereas the improvement on the reconstruction threshold is the same as in the optimal cases.

We have also obtained constructions of secret sharing schemes that are communication efficient and strongly secure at the same time. In particular, we have obtained the first secret sharing schemes with optimal communication efficiency and optimal strong security, which has the {\color{black}component-wise} security advantages of optimal perfect schemes and the {\color{black}security and} storage efficiency of optimal ramp schemes, being also communication efficient. Their field sizes are however lower bounded by (but still linear on) the lengths, as previous optimal strongly secure schemes. We have then given a construction based on algebraic geometry codes with large lengths for a fixed finite field, at the cost of near optimal communication efficiency and near optimal strong security in the same sense as in the previous paragraph.

\appendices

{\color{black}
\section{Proof of Theorem \ref{theorem strongly secure construction}} \label{app proof strongly secure}

In this appendix, we prove Theorem \ref{theorem strongly secure construction}. First, to obtain the information rate and communication overheads as in Theorem \ref{theorem general univeral staircase properties}, we need to verify that the dimensions of the codes in (\ref{eq massey type codes}) satisfy the hypotheses in Subsection \ref{subsec Construction 2}.

On the one hand, we have that 
$$ \dim(\mathcal{D}^{(1)}_{[\ell+1, \ell+n]}) = \dim(\mathcal{D}^{(1)}). $$
Therefore the projection map from any of the linear codes $ \mathcal{D}^{(j)} $ onto the coordinates in $ [\ell+1, \ell+n] $ is a vector space isomorphism. Thus
$$ k^{(j)} = \dim(\mathcal{C}^{(j)}) = \dim(\mathcal{D}^{(j)}), $$ 
for all $ j = 1,2, \ldots, h $. Now since $ \mathcal{D}_{1 [\ell]} = \mathbb{F}_q^\ell $, it holds that
$$ \ell = \dim(\mathcal{D}_{1 [\ell+1, \ell+n]}) - \dim(\mathcal{D}_{1}^{[\ell+1, \ell+n]}) = \dim(\mathcal{C}_1) - \dim(\mathcal{C}_2). $$
Hence the parameters of the codes in (\ref{eq massey type codes}) are as in Subsection \ref{subsec Construction 2}, and we only need to prove the lower bound on the strong security in Theorem \ref{theorem strongly secure construction}. It follows by combining the two following lemmas, where the first one is \cite[Theorem 19]{kurihara-secret}. We recall the proof for convenience of the reader. 

\begin{lemma}[\textbf{\cite{kurihara-secret}}] \label{lemma kurihara independent columns}
Assuming that the secret $ S \in \mathbb{F}_q^{\alpha \times \ell} $ is uniformly distributed, it holds that
$$ I(S_I; X_J) = \sum_{j=1}^u I(S_{i_j}; X_J, S_{i_1}, S_{i_2}, \ldots, S_{i_{j-1}}), $$
for all $ u = 1,2, \ldots \ell $, all $ I = \{ i_1, i_2, \ldots, i_u \} \subseteq [\ell] $ and all $ J \subseteq [n] $.
\end{lemma}
\begin{proof}
Since we are assuming that the random variable $ S $ is uniform on $ \mathbb{F}_q^{\alpha \times \ell} $, it follows that its columns are uniformly distributed in the alphabet $ \mathcal{A} = \mathbb{F}_q^\alpha $ and statistically independent. Hence the following equalities follow from the chain rule of conditional entropy (see \cite[Theorem 2.5.1]{cover}):
\begin{eqnarray} \nonumber
I(S_I; X_J) & = & H(S_{i_1}, S_{i_2}, \ldots, S_{i_u}) - H(S_{i_1}, S_{i_2}, \ldots, S_{i_u} \mid X_J) \\ \nonumber
& = & \sum_{j=1}^u H(S_{i_j}) - \sum_{j=1}^u H(S_{i_j} \mid X_J, S_{i_1}, S_{i_2}, \ldots, S_{i_{j-1}}) \\ \nonumber
& = & \sum_{j=1}^u I(S_{i_j}; X_J, S_{i_1}, S_{i_2}, \ldots, S_{i_{j-1}}). \nonumber
\end{eqnarray}
\end{proof}

\begin{lemma}
If $ i \in [\ell] $, $ I \subseteq [\ell] \setminus \{ i \} $, $ J \subseteq [n] $ and $ \# I + \# J \leq d(\mathcal{G}_i^\perp) - 1 $, then
$$ I(S_i; X_J, S_I) = 0. $$
\end{lemma}
\begin{proof}
Since $ \mathcal{D}_{1 [\ell]} = \mathbb{F}_q^\ell $, then by definition of the codes in (\ref{eq massey type codes}), we may choose a generator matrix of $ \mathcal{D}^{(1)} $ of the form
$$ F^{(1)} = \left( \begin{array}{c|c}
0 & G_2 \\
\hline
I_\ell & G_c \\
\hline
0 & G_c^{(h)} \\
\hline
\vdots & \vdots \\
\hline
0 & G_c^{(2)}
\end{array} \right) \in \mathbb{F}_q^{k^{(1)} \times (\ell + n)}, $$
where the matrices $ G_c \in \mathbb{F}_q^{\ell \times n} $, $ G_2 \in \mathbb{F}_q^{k_2 \times n} $ and $ G_c^{(j)} \in \mathbb{F}_q^{(k^{(j)} - k^{(j+1)}) \times n} $, for $ j = 2,3, \ldots, h $, are as in Subsection \ref{subsec Construction 2}.

Similarly to the proof of Theorem \ref{theorem general univeral staircase properties}, define
$$ W_{I \cup J} = (S_I, X_J) = M_h^\prime F^{(1)}_{I \cup J} \in \mathbb{F}_q^{\alpha \times \#(I \cup J)}. $$
If $ S_i \in \mathbb{F}_q^{\alpha \times 1} $ is fixed, then $ W_{I \cup J} $ is the uniform random variable in an affine space whose corresponding vector space is
$$ \mathcal{V} = \{ M_h^\prime F^{(1)}_{I \cup J} \mid S_i = 0 \} \subseteq \mathbb{F}_q^{\alpha \times \#(I \cup J)}. $$
There exist $ j_1, j_2, \ldots, j_k \in [k_1] \setminus \{ i \} $ such that the corresponding rows in
$$ \left( \begin{array}{c|c}
0 & G_2 \\
\hline
I_\ell & G_c \\
\end{array} \right)_{I \cup J} \in \mathbb{F}_q^{k_1 \times \#(I \cup J)} $$
are linearly independent and generate $ (\mathcal{G}_i)_{I \cup J} \subseteq \mathbb{F}_q^{\#(I \cup J)} $. Let $ G^\prime \in \mathbb{F}_q^{k \times \# (I \cup J)} $ be the matrix formed by such rows. Finally, define the linear map
$$ \varphi : \left( \mathcal{G}_i \right)^\alpha_{I \cup J} \longrightarrow \mathcal{V}, $$
where $ \varphi ((\boldsymbol\lambda_{j_1}, \boldsymbol\lambda_{j_2}, \ldots, \boldsymbol\lambda_{j_k}) G^\prime) $ is given by
\begin{displaymath}
\left( \boldsymbol\lambda_1, \ldots, \boldsymbol\lambda_{k_1} \left\vert
\begin{array}{cccc}
 D_{1,1} & D_{1,2} & \ldots & D_{1,h-1} \\
 D_{2,1} & D_{2,2} & \ldots & 0 \\
 D_{3,1} & D_{3,2} & \ldots & 0 \\
 \vdots & \vdots & \ddots & \vdots \\
 D_{h-1,1} & 0 & \ldots & 0 \\
 0 & 0 & \ldots & 0
\end{array} \right. \right) F^{(1)}_{I \cup J},
\end{displaymath}
where $ \boldsymbol\lambda_j \in \mathbb{F}_q^{\alpha \times 1} $, for $ j = 1,2, \ldots, k_1 $, and $ \boldsymbol\lambda_j = \mathbf{0} $ if $ j \notin \{ j_1, j_2, \ldots, j_k \} $, and where the vectors $ D_{u,v} $ are defined from $ (\boldsymbol\lambda_1, \boldsymbol\lambda_2, \ldots, \boldsymbol\lambda_{k_1}) \in \mathbb{F}_q^{\alpha \times k_1} $ as in Subsection \ref{subsec Construction 2} (that is, considering $ (R,S) = (\boldsymbol\lambda_1, \boldsymbol\lambda_2, \ldots, \boldsymbol\lambda_{k_1}) $).

First, $ \varphi $ is well-defined since $ G^\prime $ has full rank. Secondly, we see that it is one to one by an iterative reconstruction argument as in item 3 in the proof of Theorem \ref{theorem general univeral staircase properties}, due to the recursive definition of the matrices $ D_{u,v} $. We leave the details to the reader. Hence we conclude that
$$ H(W_{I \cup J} \mid S_i) = \dim(\mathcal{V}) \geq \alpha \dim((\mathcal{G}_i)_{I \cup J}). $$
Therefore, as in the proof of item 2 in Theorem \ref{theorem general univeral staircase properties}, it holds that
$$ I(S_i; W_{I \cup J}) = H(W_{I \cup J}) - H(W_{I \cup J} \mid S_i) $$
$$ \leq \alpha (\dim(\mathbb{F}_q^{\# (I \cup J)}) - \dim((\mathcal{G}_i)_{I \cup J})) = 0, $$
where the last inequality follows from Lemma \ref{lemma coset distance info}, since $ \# (I \cup J) \leq d(\mathcal{G}_i^\perp) - 1 $, and we are done.
\end{proof}

\section{Multiplicative schemes} \label{app multiplicative}

In this appendix we study parameters for which the secret sharing schemes in this paper are multiplicative or strongly multiplicative, following \cite{Cramer-Cascudo}. We will consider the Massey-type schemes from Subsection \ref{subsec massey AG codes}, and we will show that they may compute $ \alpha \ell $ products over $ \mathbb{F}_q $ or $ \alpha $ products in $ \mathbb{F}_{q^\ell} $, in parallel in both cases, while keeping their communication efficiency. It is left as open problem whether they are multiplicative over $ \mathbb{F}_{q^{\alpha \ell}} $.

Consider codes as in Corollary \ref{corollary strongly secure AG} and its proof, and let $ F : \mathcal{A}^\ell \longrightarrow \mathcal{A}^n $ be the corresponding secret sharing scheme. Let $ \ast $ be the coordinate-wise product either in $ \mathbb{F}_q^{\alpha \ell} $ or in $ \mathbb{F}_{q^\ell}^\alpha $. In the second case, we redefine the divisor $ E $ from the proof of Corollary \ref{corollary strongly secure AG} as 
$$ E = R + P_{\ell + 1} + P_{\ell + 2} + \cdots + P_{\ell + n}, $$
where $ R $ is an $ \mathbb{F}_{q^\ell} $-rational point in $ \mathcal{X} $ of degree $ \ell $ over $ \mathbb{F}_q $.

To consider multiplicative properties as in \cite{Cramer-Cascudo}, we define a new randomized function $ \widetilde{F} : \mathcal{A}^\ell \longrightarrow \mathcal{A}^n $, where if $ \mathbf{v} = \sum_i \lambda_i (\mathbf{s}_i \ast \mathbf{u}_i) $, for $ \lambda_i \in \mathbb{F}_q $ and $ \mathbf{s}_i, \mathbf{u}_i \in \mathbb{F}_{q^\ell}^\alpha $ or $ \mathbb{F}_q^{\alpha \ell} $, then
$$ \widetilde{F}(\mathbf{v}) = \sum_i \lambda_i (F(\mathbf{s}_i) \star F(\mathbf{u}_i)), $$
where $ \star $ is the coordinate-wise product in $ \mathbb{F}_q^{\alpha n} $ on the right-hand side of the last equation.

The new randomized function $ \widetilde{F} $ is obviously well-defined and linear over $ \mathbb{F}_q $ and hence can be described as in Definition \ref{definition nested coset}. Following \cite{Cramer-Cascudo}, we say that $ F $ is multiplicative (resp. $ \widetilde{r} $-strongly multiplicative) if $ \widetilde{F} $ has $ n $-reconstruction (resp. $ \widetilde{r} $-reconstruction). Equivalently, the product of two secrets can be recovered linearly by the resharing and recombination processes described in \cite[Sections 6 \& 7]{MPCfromanySS}.

Assume now that $ 2 \mu^{(1)} = 2 \delta_1 - 2 < n $. We will show that then $ \widetilde{F} $ has $ n $-reconstruction. To that end, it is enough to show that if $ \lambda_i \in \mathbb{F}_q $ and $ f_i,g_i \in \mathcal{L}(\mu^{(1)} Q) $ are such that 
\begin{equation}
f(P_{\ell + 1}) = f(P_{\ell + 2}) = \ldots = f(P_{\ell + n}) = 0,
\label{eq condition for multi}
\end{equation}
where $ f = \sum_i \lambda_i f_i g_i $, then $ f = 0 $. However, if (\ref{eq condition for multi}) holds, then 
$$ f \in \mathcal{L}(2 \mu^{(1)} Q - P_{\ell + 1} - P_{\ell + 2} - \cdots - P_{\ell + n}). $$
Since $ \deg(2 \mu^{(1)} Q - P_{\ell + 1} - P_{\ell + 2} - \cdots - P_{\ell + n}) = 2 \mu^{(1)} - n < 0 $, then $ f = 0 $ by \cite[Corollary 1.4.12]{stichtenothbook}, and we are done.

Similarly $ F $ is $ \widetilde{r} $-strongly multiplicative if $ 2 \mu^{(1)} = 2 \delta_1 - 2 < \widetilde{r} $. In the case of the optimal schemes from Corollary \ref{corollary strongly secure MDS}, we see that they are multiplicative if $ 2 \delta_1 < n $. Even though this allows for reasonable parameters, such as $ \ell = n/6 $, $ t = n/6 $, $ r = n/3 $ and $ \delta_1 = n/2 $, it would be desirable to obtain multiplicative schemes with the only constraint $ 2r < n $ (as in \cite{Cramer-Cascudo}), but arbitrary $ \Delta \subseteq [r,n ] $. Overcoming this issue is left as open problem for future research.

}

\appendices

\section*{Acknowledgement}

The author is thankful for the support and guidance of his advisors Olav Geil and Diego Ruano. At the time of submission, the author was visiting the the Edward S. Rogers Sr. Department of Electrical and Computer Engineering, University of Toronto. He greatly appreciates the support and hospitality of Frank R. Kschischang. {\color{black}Finally, the author wishes to thank the anonymous reviewers for their very helpful comments, which greatly improved this paper. }

\ifCLASSOPTIONcaptionsoff
  \newpage
\fi



\bibliographystyle{IEEEtranS}
\end{document}